\documentclass{llncs}
\newcounter{bibcounter}
\usepackage{amssymb,amsmath,graphics,color,enumerate}
\usepackage{graphicx}
\usepackage{multicol}
\newcommand{\NP}{{\sf NP}}
\newcommand{\FPT}{{\sf FPT}}

\newcommand{\W}{{\sf W}}

\newcommand{\F}{\mathcal{F}}

\usepackage{pgf,tikz}
\usetikzlibrary{arrows,shapes}

\usepackage{comment}

\spnewtheorem*{Proof}{Proof}{\itshape}{\rmfamily}
\renewenvironment{proof}{\begin{Proof}}{\qed\end{Proof}}

\begin{document}
\title{Finding Shortest Paths between\\ Graph Colourings\thanks{Supported by EPSRC (EP/G043434/1), by a Scheme 7 grant from the London Mathematical Society, and by the German Research Foundation (KR 4286/1).}
}
\author{Matthew Johnson\inst{1}\and Dieter Kratsch\inst{2}\and Stefan Kratsch\inst{3} \and\\ Viresh Patel\inst{4} \and Dani\"el Paulusma\inst{1}
}
 \institute{
 School of Engineering and  Computing Sciences, Durham University,\\
 Science Laboratories, South Road,
 Durham DH1 3LE, United Kingdom
 \texttt{\{matthew.johnson2,daniel.paulusma\}@durham.ac.uk}
 \and
 Laboratoire d'Informatique Th\'eorique et Appliqu\'ee, Universit\'e de Lorraine,\\ 57045 Metz Cedex 01, France, \texttt{dieter.kratsch@univ-lorraine.fr}
 \and
 Institut f\"ur Softwaretechnik und Theoretische Informatik,\\
 Technische Universit\"at Berlin, Germany, \texttt{stefan.kratsch@tu-berlin.de}
 \and
 School of Mathematical Sciences,
 Queen Mary, University of London,\\
 Mile End Road, London E1 4NS, United Kingdom
 \texttt{viresh.patel@qmul.ac.uk}
 }
\maketitle

\begin{abstract}
The $k$-colouring reconfiguration problem asks whether, for a given graph $G$, two proper $k$-colourings $\alpha$ and $\beta$ of $G$, and a positive integer $\ell$, there exists a sequence of at most $\ell+1$ proper $k$-colourings of $G$ which starts with $\alpha$ and ends with $\beta$ and where successive colourings in the sequence differ on exactly one vertex of $G$. 
We give a complete picture of the parameterized complexity of the $k$-colouring reconfiguration problem for each fixed $k$ when parameterized by $\ell$.
First we show that the $k$-colouring reconfiguration problem is polynomial-time solvable for $k=3$, settling an 
open problem of 
Cereceda, van den Heuvel and Johnson.
Then,  for all $k \geq 4$, we show that the $k$-colouring reconfiguration problem, when parameterized by 
$\ell$, is fixed-parameter tractable  (addressing a question of Mouawad, Nishimura, Raman, Simjour and Suzuki) but that it has no polynomial kernel unless the polynomial hierarchy collapses.
\end{abstract}

\section{Introduction}\label{sec:intro}

Graph colouring has its origin in a nineteenth century map colouring problem and has now been an active area of research for more than 150 years, finding many applications within and beyond Computer Science and Mathematics. Given a graph $G=(V,E)$ and a positive integer $k$, 
a {\it $k$-colouring} of $G$ is a map $c\colon V \rightarrow \{1, \ldots, k\}$; it is {\it proper} if $c(u)\neq c(v)$ for all $u,v$ with $uv\in E$. 
The problem of deciding whether a graph has a proper
$k$-colouring for fixed $k \geq 3$ was an early example of an \NP-complete problem~\cite{Lo73}.
If, however, one knows that a graph has a proper 
$k$-colouring, or several of them,
one may wish to know more about them such as how many there are or what structural properties they have. 

One way to study these questions is to consider the $k$-colouring reconfiguration graph: given a graph $G$, the $k$-colouring reconfiguration graph $R_k(G)$ of $G$ is a graph whose vertices are the proper $k$-colourings of $G$ and where an edge is present between two $k$-colourings if and only if the two $k$-colourings differ on only a single vertex of $G$.  

There are several algorithmic questions one can ask about the graph $R_k(G)$ such as whether $R_k(G)$ is connected, whether there exists a path between two given vertices of $R_k(G)$, or how long is the shortest path between two given vertices of $R_k(G)$. (Note that in general $R_k(G)$ has size exponential in the size of $G$, making these questions highly non-trivial.) It is the latter question, stated formally below, that we address in this paper. 
  
 \medskip
\noindent
\textsc{$k$-Colouring Reconfiguration}

\medskip
\noindent
\begin{tabular}{p{1.7cm}p{10cm}}
\textit{Instance}\,:& An $n$-vertex graph $G=(V,E)$, two 
proper  $k$-colourings $\alpha$ and $\beta$ and a positive integer~$\ell$.\\
\textit{Question}\,:&Is there a path in the reconfiguration graph of $G$ between $\alpha$ and 
$\beta$ of length at most $\ell$? 
\end{tabular}

\subsection{General Motivation}
Reconfiguration graphs can be defined for any search problem: the vertices correspond to all solutions to the problem 
and the edges are defined by a symmetric adjacency relation  normally chosen to represent a smallest possible change between solutions. They arise naturally when one wishes to understand the solution space for a search problem.

There has been much research over the last ten years on the structure and algorithmic aspects of reconfiguration graphs, not only for 
{\sc $k$-Colouring}~\cite{BB13,BJLPP14,BC09,CHJ06,CHJ06a,CHJ06b} 
but also for many other  problems, such as 
{\sc Satisfiability}~\cite{GKMP06},
{\sc Independent Set}~\cite{Bo14,BKW14,KMM12}, 
{\sc List Edge Colouring}~\cite{IKD09,IKZ11},
{\sc $L(2,1)$-Labeling}~\cite{IKOZ12},
{\sc Shortest Path}~\cite{Bo10,Bo13,KMM11}, 
and {\sc Subset Sum}~\cite{ID11}. 
From these studies, the following subtle phenomenon has been observed, which one would like to better understand: it is often (but not always) the case that \NP-complete search problems give rise to PSPACE-complete\footnote{PSPACE-completeness appears to be the default complexity for intractable instances of this kind of problem; see~\cite{IDHPSUU10}.} reconfiguration problems, whereas polynomial-time solvable search problems often give rise to polynomial-time solvable reconfiguration problems. 
For further background we refer the reader to the recent survey of  van den Heuvel~\cite{He13}.

Reconfiguration graphs are also important for constructing and analyzing algorithms that sample or count solutions to a search problem. Indeed, understanding connectivity properties of the $k$-colouring reconfiguration graph is fundamental in  analyzing certain randomized algorithms for sampling and counting $k$-colourings of a graph and in analyzing certain cases of the Glauber dynamics in statistical physics (see Section~5 of \cite{He13}).

\subsection{Our Results}
Our first result, 
which we prove in Section~\ref{s-poly}, 
shows that \textsc{$k$-Colouring Reconfiguration} can be solved in polynomial time when $k=3$, which settles a problem raised by Cereceda, van den Heuvel and Johnson~\cite{CHJ06b}. Note that the cases $k=1,2$ are easily seen to be polynomial-time solvable.

In \cite{CHJ06b}, Cereceda et al.\ were mainly concerned with determining whether, given a graph $G$ on $n$ vertices and $m$ edges, and two 
proper $3$-colourings $\alpha$ and $\beta$, there exists \emph{any} path between $\alpha$ and $\beta$ in $R_k(G)$. They found a polynomial-time algorithm to solve this problem and further showed that, for certain instances, their algorithm in fact finds a shortest path between $\alpha$ and $\beta$ (a precise statement is given in Section~\ref{s-poly}). Here we complete their result by giving an algorithm for all instances.

\begin{theorem}\label{t-ptime}
{\sc $3$-Colouring Reconfiguration} can be solved in time $O(n+m)$. 
\end{theorem}

For $k \geq 4$, we cannot expect a polynomial-time algorithm for {\sc $k$-Colouring Reconfiguration}: 
Bonsma and Cereceda~\cite{BC09} showed that, for each~$k \geq 4$, the problem of determining if there is \emph{any} path between two given 
proper $k$-colourings of a given graph is PSPACE-complete. On the other hand, our second result (proven in Section~\ref{s-fpt}) is that for each $k \geq 4$, {\sc $k$-Colouring Reconfiguration} is fixed-parameter tractable when parameterized by the path length~$\ell$.

Recall that, informally, a parameterized problem is a decision problem (in our case {\sc $k$-Colouring Reconfiguration}) in which every problem instance~$I$ has an associated integer parameter $p$ (in our case the path length $\ell$). A parameterized problem is \emph{fixed-parameter tractable} (\FPT) if every instance~$I$ can be solved in time $f(p)|I|^c$ where $f$ is a computable function that only depends on~$p$ and $c$ is a constant independent of $p$ (see, for example,~Niedermeier~\cite{Ni06} for an overview).

\begin{theorem} \label{t-fpt}
For each fixed $k \geq 4$, {\sc $k$-Colouring Reconfiguration} can be solved in time $O((k\cdot \ell)^{\ell^2 +\ell} \cdot \ell n^2)$. In particular, for each fixed $k \geq 4$, {\sc $k$-Colouring Reconfiguration} is \FPT\ when parameterized by~$\ell$.
\end{theorem}

Once a problem is shown to be \FPT\ (and it is unlikely that the problem is polynomial-time solvable), one can go further and ask whether it has a {\it polynomial kernel}.  It is well known~\cite{Ni06} that a problem is \FPT\ with respect to a parameter~$p$ if and only if it can be {\it kernelized}, i.e., if and only if, for any instance $(I, p)$ of the given parameterized problem, it is possible to compute in polynomial time an {\it equivalent instance} $(I', p')$ such that $|I'|,p'\leq g(p)$ for some computable 
function~$g$ (two problem instances are equivalent if and only if they are both yes-instances or both no-instances). If $g(p)$ is a polynomial, then the given parameterized problem is said to have a polynomial kernel. We prove the following theorem in Section~\ref{s-nokernel}.
\begin{theorem} \label{t-ker}
For each fixed $k\geq 4$,
$k$-\textsc{Colouring Reconfiguration} 
parameterized by $\ell$
does not admit a polynomial kernel  
unless $\mathsf{NP\subseteq coNP/poly}$.
\end{theorem}
In fact Theorem~\ref{t-ker} holds even when we restrict attention to 
inputs where the two proper $k$-colourings of the input graph differ in only two vertices (note that the problem becomes trivial if the two given $k$-colourings differ in only one vertex). 

Our three results give a {\it complete} picture of the parameterized complexity of {\sc $k$-Colouring Reconfiguration} for each fixed $k$ when parameterized by $\ell$.

\subsection{Related work}
Fixed-parameter tractability of \textsc{$k$-Colouring Reconfiguration} was proved independently in recent work of Bonsma and Mouawad~\cite{BM14}  (their algorithm is different from ours). They also prove various hardness results for other parameterizations of {\sc $k$-Co\-lou\-ring Reconfiguration}.  In particular, they proved that if $k$ is part of the input then {\sc $k$-Colouring Reconfiguration} is \W[1]-hard when parameterized only by~$\ell$ (note that the problem, when parameterized only by~$k$, is para-PSPACE-complete due to the aforementioned result of Bonsma and Cereceda~\cite{BC09}).

Mouawad, Nishimura, Raman, Simjour and Suzuki~\cite{MNRSS13} were the first to consider reconfiguration problems in the context of parameterized complexity. For various \NP-complete search problems, they showed that determining whether there exists a path of length at most $\ell$ in the reconfiguration graph between two given vertices is \W[1]-hard (when $\ell$ is the parameter); they asked if there exists an \NP-complete problem for which the corresponding reconfiguration problem, parameterized by $\ell$, is \FPT.  Theorem~\ref{t-fpt} and~\cite{BM14} give the second positive answer to this question, the first being an \FPT\ algorithm for a reconfiguration problem related to {\sc Vertex Cover} \cite{MNR14}.  However, perhaps surprisingly, Theorem~\ref{t-ptime} shows that there even exists an \NP-complete problem for which the corresponding shortest path problem in the reconfiguration graph is polynomial-time solvable, and thus trivially \FPT\ when parameterized by~$\ell$. 

As mentioned earlier, deciding whether there exists \emph{any} path in $R_k(G)$ between two $k$-colourings $\alpha$ and $\beta$ of an input graph $G$ is polynomial-time solvable for $k \leq 3$ \cite{CHJ06b} and PSPACE-complete for $k \geq 4$ \cite{BC09}. The problem remains PSPACE-complete for bipartite graphs when $k \geq 4$, for planar graphs when $4\leq k \leq 6$ and for planar bipartite graphs for $k=4$~\cite{BC09}.

The algorithmic question of whether $R_k(G)$ is connected for a given $G$ is addressed in \cite{CHJ06,CHJ06a}, where it is shown that the problem is coNP-complete for $k=3$ and bipartite $G$, but polynomial-time solvable for planar bipartite $G$.

Finally, the study of the diameter of $R_k(G)$ raises interesting questions. In~\cite{CHJ06b} it is shown that every component of $R_3(G)$ has diameter polynomial (in fact quadratic) in the size of $G$. On the other hand, for $k \geq 4$, explicit constructions~\cite{BC09} are given of graphs  $G$ for which $R_k(G)$ has at least one component with diameter exponential in the size of $G$. It is known that if $G$ is a $(k-2)$-degenerate graph then $R_k(G)$ is connected and it is conjectured that in this case $R_k(G)$ has diameter polynomial in the size of $G$ \cite{CHJ06}; for graphs of treewidth $k-2$ the conjecture has been proved in the 
affirmative~\cite{BB13}. 

\subsection{Definitions and Terminology}

Here are some definitions needed throughout the paper. Let $G=(V,E)$ be a graph on $n$ vertices and $m$ edges.  For \emph{any} two colourings  $c$ and $d$, we say that $c$ and $d$ \emph{agree} on a vertex $u$ if $c(u)=d(u)$ and that otherwise they \emph{disagree} on $u$.  A \emph{$(c_0\!\rightarrow\!\!c_\ell)$-recolouring} of $G$ of {\em length} $\ell$ is a sequence $R = c_0, \dots , c_{\ell}$ of proper colourings of $G$,  where for $1 \leq q \leq \ell$, $c_q$ and $c_{q-1}$ disagree on at most one vertex. So possibly $c_q=c_{q-1}$ though in this case $c_q$ could be deleted and the sequence that remained would also be an $(c_0\!\rightarrow\!\!c_\ell)$-recolouring (of length $\ell - 1$).  The set $\{c_{q-1}c_q : c_{q-1} \neq c_q\}$ is a set of edges in the reconfiguration graph corresponding to a walk from $c_0$ to~$c_\ell$.  An edge of the reconfiguration graph corresponds to the \emph{recolouring} of a single vertex, and we say that it is possible to recolour a vertex $v$, with respect to a colouring $c$, if there is a colour other than $c(v)$ that does not appear on any neighbour of $v$.

\section{A Polynomial-Time Algorithm for $k=3$}\label{s-poly}

We consider {\sc 3-Colouring Reconfiguration}.
Throughout this section $G=(V,E)$ is a 
3-colourable graph 
with $n$ vertices and $m$ edges, $\alpha$ and $\beta$ are two of its 3-colourings, and $\ell$ is a positive integer.
We will prove Theorem~\ref{t-ptime} by showing that we can decide in time 
$O(n+m)$ whether there is a path in $R_3(G)$ from $\alpha$ to $\beta$ of length at most $\ell$.    Cereceda et al.~\cite{CHJ06b} provided a partial solution to {\sc 3-Colouring Reconfiguration}.  They introduced a polynomial-time algorithm that determines whether or not $\alpha$ and $\beta$ belong to the same component of $R_3(G)$, and, when they are connected, finds a path between them of  length $O(n^2)$.  They noted that in some, but not all, cases the path found by their algorithm is shortest possible.  Our approach here extends the techniques they introduced.

We describe the section in outline.  We will introduce two features of 3-coloured graphs: the set of fixed vertices and a height function for the vertex set.  In Lemmas~\ref{l-Gfc} and~\ref{l-heights-nec}, we will prove necessary conditions, in terms of fixed vertices and vertex heights, for $\alpha$ and $\beta$ to belong to the same component of the reconfiguration graph, and, in Lemmas~\ref{l-fixed-time} and~\ref{l-heights-time} show that these conditions can be checked in time 
$O(n+m)$.  We will then show that these conditions are sufficient by describing an algorithm that finds a path between $\alpha$ and $\beta$ whenever the conditions are satisfied.  Furthermore in Lemma~\ref{l-k-bound} we will give a lower bound on the length of the shortest path from $\alpha$ to $\beta$ in terms of vertex heights, and we will show that the path found by our algorithm achieves that bound.  
In Lemma~\ref{l-k-time}, we show that the lower bound can be computed in time $O(n+m)$, and so this is sufficient time to decide whether or not there is a path from $\alpha$ to $\beta$ of length at most $\ell$, proving Theorem~\ref{t-ptime}.

\subsection{Fixed Vertices}
Let $G=(V,E)$ be a graph and let $c$ be a $3$-colouring of $G$ with colours $1$, $2$, and $3$.  
As in \cite[Claim~8]{CHJ06b}, we define a vertex $v$ to be a \emph{fixed vertex of $G$ (with respect to $c$)} if there is no sequence of recolourings from $c$ that will allow us to recolour $v$, i.e.\ for every $3$-colouring $c'$ that is in the same component of $R_3(G)$ as $c$, we have $c(v) = c'(v)$. For example, if a cycle with $0 \bmod 3$ vertices is coloured $123123 \cdots 123$ by $c$, then every vertex on the cycle is fixed (as none can
be the first to be recoloured): we call such a cycle a \emph{fixed cycle} (as a subgraph of $G$, and with respect to the $3$-colouring $c$). The set of fixed vertices of $G$ with respect to $c$ is denoted $F_{G,c}$ and $F_{G,c}^i \subseteq F_{G,c}$ denotes the set of fixed vertices coloured $i$.

The next lemma follows immediately from the definition of fixed vertices, but we state it formally so that we can refer to it later.  
\begin{lemma} \label{l-Gfc} 
Let $G$ be a graph and let $\alpha$ and $\beta$ be two $3$-colourings of $G$.  If $\alpha$ and $\beta$ belong to the same component of $R_3(G)$, then $F_{G,\alpha}^i=F_{G,\beta}^i$ for $i=1,2,3$. 
\end{lemma}

\begin{lemma} \label{l-fixed-time}
Let $G=(V,E)$ be a graph and let $c$ be a $3$-colouring of $G$.
The set of fixed vertices $F_{G,c}$ can be found in time 
$O(n+m)$.
\end{lemma}

\begin{proof}
We present an algorithm to find $F_{G,c}$.  

Let $S$ be set initially to be $V$.  The algorithm will delete vertices from $S$ and we will claim that when the algorithm terminates it is equal to $F_{G,c}$.   For $i=1,2,3$, for each $v \in V$, let $n_i(v)$ be $-1$ if $c(v)=i$ or otherwise the number of neighbours of $v$ in $G[S]$ coloured $i$.  Mark each vertex $v$, for which $n_i(v)=0$ for some $i$, as \emph{waiting}.   This initialization can be done in time $O(n+m)$.

The algorithm repeatedly chooses a waiting vertex $v$ (at which point it is \emph{processed} and no longer waiting).  Then $v$ is deleted from $S$, and so, for each of its neighbours $w$, $n_i(w)$ is decremented by 1, where $i=c(v)$.  If $n_i(w)=0$, and $w$ is not processed, it is marked as waiting.  The algorithm terminates when no vertex is waiting.  Within all the executions of this loop, each edge is considered at most twice so again time $O(n+m)$ suffices.

Consider $S$ when the algorithm terminates.
Each vertex $v$ in $S$  has not been processed and so $n_i(v) \neq 0$ for each $i=1,2,3$.  Thus, for this terminal $S$, every vertex in $G[S]$ has at least one neighbour of each of the two colours distinct from $c(v)$.  This implies that these vertices are fixed since none can be recoloured until at least one of the others has been recoloured.

For each vertex $v$ that is deleted from $S$,  let $d(v)$ be a colour $i$ such that $n_i(v)=0$ when $v$ is deleted. Note that this implies $d(v) \neq c(v)$ since $n_{c(v)}(v)=-1$.

To complete the proof, we show that a vertex $u$ that has been deleted from $S$ can be recoloured.  Consider the vertices in the order of their deletion from $S$ up to and including $u$:  for each vertex $v$ with $d(v) - c(v) \equiv d(u) - c(u) \bmod 3$, recolour $v$ with $d(v)$.  Thus $u$ is ultimately recoloured and we need only show that when a vertex $v$ is recoloured with $d(v)$ each neighbour $w$ with $c(w)=d(v)$ has already been recoloured and each neighbour with $d(w)=d(v)$ has not been recoloured.   Without loss of generality, suppose that $c(v)=1$ and that $d(v)=2$.  Using the definition of $d(v)$ twice, any neighbour $w$ of $v$ with $c(w)=2$ must have been deleted from $S$ before $v$, and $d(w)$ must be $3$.  So $w$ will be recoloured from 2 to 3 before $v$ is recoloured from $1$ to $2$ since $d(v)-c(v) \equiv d(w)-c(w)$.  If a neighbour $w$ of $v$ has $d(w)=2$, then $c(w)=3$ (else $c(v)=c(w)$) and so $w$ will not be recoloured as  $d(v) - c(v) \not \equiv d(w) - c(w) \bmod 3$.
\end{proof}

\subsection{Vertex Heights}

In this subsection we define the height of a vertex with respect to a $3$-colouring and prove some properties of this function. The notion of height is similar to that in \cite{CHJ06b}.
Before heights, we have weights. Let $G$ be a graph and $c$ a $3$-colouring of~$G$.  The \emph{weight} of an edge oriented from~$u$ to~$v$ is a value $w(c,\overrightarrow{uv}) \in \{-1, 1\}$ such that $w(c,\overrightarrow{uv}) \equiv c(v) - c(u) \bmod 3$.  That is, we think of the colours cyclically: the weight of the edge reflects whether the colour is being increased or decreased $\bmod 3$ as the oriented edge is traversed.
Note that we always have $w(c,\overrightarrow{uv}) = -w(c,\overrightarrow{vu})$.

To orient a path $P$ or cycle $C$ of $G$ is to orient each edge so that a directed path $\overrightarrow{P}$  or cycle $\overrightarrow{C}$ is obtained.  The weight of an oriented path $w(c,\overrightarrow{P})$ or an oriented cycle $w(c,\overrightarrow{C})$ is the sum of the weights of its edges.   The following lemma is from~\cite{CHJ06,CHJ06a}.

\begin{lemma} \label{l-fixed-weight-cycles}
Let $G$ be a graph and let $c$ and $d$ be two $3$-colourings of $G$.  If $c$ and $d$ belong to the same component of $R_3(G)$, then, for every cycle $C$ of $G$,   $w(c,\overrightarrow{C}) = w(d,\overrightarrow{C})$.
\end{lemma}

Let $u$ be a given vertex of a connected graph $G=(V,E)$ and $\alpha$ a given $3$-colouring of $G$. Let $T$ be a spanning tree of $G$.  For any vertex $v \in V$, let $\overrightarrow{P_{uv}}$ be the (unique) oriented path from $u$ to $v$ in $T$.   For any  $3$-colouring $c$ of $G$, the relative height of $c,v$ (with respect to $\alpha$ and $u$) is $h_{\alpha, u}(c, v)$, where
\[
h_{\alpha,u}(c, v)=w(c,\overrightarrow{P_{uv}})-w(\alpha,\overrightarrow{P_{uv}}).
\]
The next result says that the particular choice of $u$ in the definition above is not too important, although it will play a role later.
\begin{lemma}
\label{l-basechange}
Let $u,u'$ be two given vertices of a connected graph $G=(V,E)$ and let $\alpha, \beta$ be two given $3$-colourings of $G$. There exists a constant $D$ such that for every $v \in V$, we have
\[
h_{\alpha,u}(\beta, v) = h_{\alpha,u'}(\beta, v) + D. 
\]
\end{lemma}
\begin{proof}
We have for all $v \in V$ that
\begin{align*}
h_{\alpha,u}(\beta, v) - h_{\alpha,u'}(\beta, v)
&=w(\beta,\overrightarrow{P_{uv}})-w(\alpha,\overrightarrow{P_{uv}}) 
-w(\beta,\overrightarrow{P_{u'v}})+w(\alpha,\overrightarrow{P_{u'v}}) \\
&= w(\beta,\overrightarrow{P_{uu'}})-w(\alpha,\overrightarrow{P_{uu'}}) =:D,
\end{align*}
The second equality follows because when evaluating $w(\cdot,\overrightarrow{P_{uv}})
-w(\cdot,\overrightarrow{P_{u'v}})$, we are summing the weights of edges in a walk on $T$ which consists of $\overrightarrow{P_{uv}}$ concatenated with $\overrightarrow{P_{vu'}}$ and after cancellation of edges traversed in opposite directions, we are left with $\overrightarrow{P_{uu'}}$. 
\end{proof}

The next lemma shows how relative heights change when the vertex in the second argument changes.
\begin{lemma} \label{l-heights-nec}
Let $G=(V,E)$ be a connected graph and let $\alpha$ and $c$ be two $3$-colourings of $G$.  If $\alpha$ and $c$ belong to the same component of $R_3(G)$, then, for each $vw \in E$,
\begin{equation} \label{eq-rel-heights}
h_{\alpha,u}(c, v)-h_{\alpha,u}(c, w)+w(c,\overrightarrow{vw}) = w(\alpha,\overrightarrow{vw}). 
\end{equation}
\end{lemma}
\begin{proof}
If $vw \in T$, the lemma follows from the definition of relative height.
Otherwise let $x$ be the vertex farthest from $u$ that lies on both $\overrightarrow{P_{uv}}$ and $\overrightarrow{P_{uw}}$.
So,  for any colouring, 
\begin{eqnarray} 
w(\cdot, \overrightarrow{P_{xv}}) & = & w(\cdot, \overrightarrow{P_{uv}}) - w(\cdot,\overrightarrow{P_{ux}}),  \label{eq-paths1} 
\\
w(\cdot, \overrightarrow{P_{xw}}) & = & w(\cdot, \overrightarrow{P_{uw}}) - w(\cdot,\overrightarrow{P_{ux}}).  \label{eq-paths2}
\end{eqnarray}
  Notice that $\overrightarrow{P_{xv}}$, $\overrightarrow{vw}$ and $\overleftarrow{P_{xw}}$ form an oriented cycle so, from Lemma~\ref{l-fixed-weight-cycles}, we have
\begin{equation*}
w(\alpha,\overrightarrow{P_{xv}})+ w(\alpha,\overrightarrow{vw})- w(\alpha, \overrightarrow{P_{xw}} )
 = w(c,\overrightarrow{P_{xv}})+ w(c,\overrightarrow{vw})- w(c, \overrightarrow{P_{xw}}),
\end{equation*}
and substituting~(\ref{eq-paths1}) and~(\ref{eq-paths2}) (with $\cdot = \alpha,c$), and cancelling terms we obtain
\begin{equation*}
w(\alpha,\overrightarrow{P_{uv}})+ w(\alpha,\overrightarrow{vw})- w(\alpha, \overrightarrow{P_{uw}} )
 = w(c,\overrightarrow{P_{uv}})+ w(c,\overrightarrow{vw})- w(c, \overrightarrow{P_{uw}}).
\end{equation*}
Rearranging and applying the definition of relative height shows that~(\ref{eq-rel-heights}) holds.
\end{proof}

Next we show we can compute relative heights efficiently.
\begin{lemma} \label{l-heights-time}
Let $G$ be connected a graph, let $\alpha$ be a given $3$-colouring of $G$ and let $u$ be a given vertex of $G$. For any $3$-colouring $\beta$  of $G$ it is possible to find $h_{\alpha,u}(\beta,v)$ for every vertex $v$ of $G$ in time $O(n)$.
\end{lemma}
\begin{proof}
From the definition, we need only find the weight of each oriented path from $u$ to $v$ in $T$ for both $\alpha$ and $\beta$.  
This can be done with two breadth-first searches on $T$ in time $O(n)$.
\end{proof}

\subsection{Recolouring: Changing the Height of a Vertex}

In order to understand the role of vertex heights in finding shortest paths between colourings, we investigate how the heights of vertices change along  the colourings in a recolouring sequence. Let $G$ be a
connected graph, let $\alpha$ and $\beta$ be two $3$-colourings of $G$ and let $R = c_0,c_2, \ldots, c_{\ell}$ be an $(\alpha\! \rightarrow\!\!\beta)$-recolouring. For a given vertex $u$ of $G$, we call $H_u^R(c_i, v)$  the \emph{absolute height} of vertex $v$ (with respect to $u,R,c_i$) and define it as follows:
\begin{itemize}
\item $H_u^R(c_0, u)=0$;
\item For $i>0$, 
\begin{equation} \nonumber
H_u^R(c_i, u) = \left\{ 
\begin{array}{ll}
        H_u^R(c_{i-1}, u),&\text{if $c_i(u)= c_{i-1}(u)$};\\
		H_u^R(c_{i-1}, u)+2,&\text{if $c_i(u)\equiv c_{i-1}(u)-1 \bmod 3$};\\
		H_u^R(c_{i-1}, u)-2,&\text{if $c_i(u)\equiv c_{i-1}(u)+1 \bmod 3$};
        \end{array}\right.
\end{equation}
\item For $i \geq 0$, $H_u^R(c_i, v) = H_u^R(c_i, u) + h_{\alpha,u}(c_i, v)$.
\end{itemize}
Let us elaborate on this a little and see how heights change along the recolouring sequence.  We define the height of $u$ to be initially zero.
If its colour is increased (decreased), then the height is decreased (increased) by 2.  We use an example to explain the motivation of the definition.  If the colour of $u$ is increased from 1 to 2, then all the neighbours of $u$ are coloured 3.  When $u$ is coloured 1, all the edges from its neighbours (oriented towards $u$) have weight $+1$; in some sense, $u$ is sitting above all its neighbours.  When $u$ is coloured 2, all the edges from its neighbours have weight $-1$, and we think of $u$ as being below them.  So the increase in colour corresponds to a reduction in height.  Similarly if the colour decreases, the height is raised.

The next lemma tells us that as other vertices are recoloured their heights change in the same way as $u$.
\begin{lemma} \label{l-absolute}
Let $G=(V,E)$ be a connected graph, let $\alpha$ and $\beta$ be two $3$-colourings of $G$, let  $R=c_0, \dots , c_{\ell}$  be an $(\alpha\!\rightarrow\!\!\beta)$-recolouring and let $u$ be a given vertex of $G$.
For each $v \in V$, $H_u^R(c_0, v)=0$.  For $i>0$, for  each $v \in V$:
\begin{equation} \nonumber
H_u^R(c_i, v) - H_u^R(c_{i-1}, v) = \left\{ 
\begin{array}{ll}
        0,&\text{if $c_i(v)= c_{i-1}(v)$};\\
		2,&\text{if $c_i(v)\equiv c_{i-1}(v)-1 \bmod 3$};\\
		-2,&\text{if $c_i(v)\equiv c_{i-1}(v)+1 \bmod 3$}.
        \end{array}\right.
\end{equation}
\end{lemma}

\begin{proof}
We prove the lemma by induction on the number of vertices. To ease notation, define $D_i(v) = H_u^R(c_i, v) - H_u^R(c_{i-1}, v)$.

Recall that we define $h_{\alpha,u}$ and hence $H_u^R$ by first fixing a spanning tree $T$, and for every vertex $v \in V$, we define $\overrightarrow{P_{uv}}$ to be the directed path in $T$ from 
$u$ to $v$. Let $v_1, \ldots, v_n$ be a breadth-first ordering of the vertices of $T$ with $v_1 = u$.

From the definition of $H_u^R$, we see that for all $i$, $D_i(v)$ satisfies the conclusion of the lemma when $v=v_1 = u$. Assume that $D_i(v)$ satisfies the conclusion of the lemma for all $i$ and all $v \in \{v_1, \ldots, v_{k-1} \}$. Observe that if $v_{k^*}$ is the ancestor of $v_k$ in $T$ (so $k^* <k$) then for all $i$, we have
\begin{align*}
H_u^R(c_i,v_k) 
&=H_u^R(c_i, u ) + h_{\alpha, u}(c_i, v_k) \\
&=H_u^R(c_i, u) + h_{\alpha, u}(c_i, v_{k^*}) + w(c_i, \overrightarrow{v_{k^*}v_k}) - w(\alpha, \overrightarrow{v_{k^*}v_k}) \\
&=H_u^R(c_i, v_{k^*}) + w(c_i, \overrightarrow{v_{k^*}v_k}) - w(\alpha, \overrightarrow{v_{k^*}v_k}),
\end{align*}
and hence
\begin{equation}
\label{eq:D}
D_i(v_k) = D_i(v_{k^*}) 
+ \left[ w(c_i, \overrightarrow{v_{k^*}v_k})
- w(c_{i-1}, \overrightarrow{v_{k^*}v_k}) \right].
\end{equation}
Now let us check that for each $i$, $D_i(v_k)$ satisfies the conclusion of the lemma given that for all $i$, $D_i(v_{k^*})$ satisfies the conclusion of the lemma.

Suppose $c_{i-1}(v_{k^*}) = c_i(v_{k^*})$ and $c_{i-1}(v_{k}) = c_i(v_{k})$. Then both terms on the RHS of (\ref{eq:D}) are zero, and so $D_i(v_k)=0$ as required.

Suppose $c_{i-1}(v_{k^*}) \not= c_i(v_{k^*})$ (and hence $c_{i-1}(v_{k}) = c_i(v_{k})$). Then both terms on the RHS of (\ref{eq:D}) are non-zero, but they cancel and so $D_i(v_k)=0$ as required.
 (For example, if the colour of $v_{k^*}$ increases by $1$ (say from $2$ to $3$) then $D_i(v_{k^*})=-2$ by induction, and we must have $c_{i-1}(v_k) = c_i(v_k) = 1$, from which we get that $w(c_i, \overrightarrow{v_{k^*}v_k})
- w(c_{i-1}, \overrightarrow{v_{k^*}v_k}) = 2$ as required. All other cases follow similarly.)

Finally suppose $c_{i-1}(v_{k}) \not= c_i(v_{k})$ (and hence $c_{i-1}(v_{k^*}) = c_i(v_{k^*})$). Then, by induction, $D_i(v_{k^*})=0$, and $w(c_i, \overrightarrow{v_{k^*}v_k})
- w(c_{i-1}, \overrightarrow{v_{k^*}v_k})$ takes exactly the value required. 
(For example, if the colour of $v_{k}$ decreases by $1$ (say from $3$ to $2$) then  we must have $c_{i-1}(v_{k^*}) = c_i(v_{k^*}) = 1$, from which we get that $w(c_i, \overrightarrow{v_{k^*}v_k})
- w(c_{i-1}, \overrightarrow{v_{k^*}v_k}) = 2$ as required. All other cases follow similarly.) 
\end{proof}

In summary: for an $(\alpha\!\rightarrow\!\!\beta)$-recolouring $R$, we have defined an absolute height function $H_u^R$ such that the height of each vertex is initially 
zero and changes by $+2$ or $-2$ whenever its colour is, respectively, decreased or increased $\bmod 3$.

\subsection{Total Heights}

In this subsection we obtain a lower bound on the length of an $(\alpha\!\rightarrow\!\!\beta)$-recolouring in terms of vertex heights.  We also see that the value of the lower bound can be found in time $O(n+m)$, and in the next subsection we will show that a $(\alpha\!\rightarrow\!\!\beta)$-recolouring that achieves the bound can be found.

Let $G$ be a connected graph, let $\alpha$ and $\beta$ be two $3$-colourings of $G$ and let $R = c_0, \ldots, c_{\ell}$ be an $(\alpha\!\rightarrow\!\!\beta)$-recolouring. For a given vertex $u$ of $G$, the \emph{total height of $R$ (with respect to $u$)} is given by

\[ T_u(R) = \sum_{v \in V} |H_u^R(\beta, v)| = \sum_{v \in V}  \left|H_u^R(\beta, u) + h_{\alpha, u}(\beta, v)\right| 
\]

\begin{lemma} \label{l-totalheights}
Let $G=(V,E)$ be a connected graph with given vertex $u$ and let $\alpha$ and $\beta$ be two $3$-colourings of $G$.
Let $R$ be a $(\alpha\!\rightarrow\!\!\beta)$-recolouring of length $\ell$.  Then $\ell \geq \frac{1}{2} T_u(R)$.
\end{lemma}
\begin{proof}
We know that for all $v \in V$, $H_u^R(\alpha, v)=0$ by Lemma~\ref{l-absolute}.  For each colouring in $R$, the absolute height of only one vertex differs from the previous colouring in $R$ and the height difference is 2, again by Lemma~\ref{l-absolute}.  Thus, each vertex $v$ must change colour at least $|H_u^R(\beta,v)|/2$ times in $R$ and so the total number of colourings in $R$ must be at least $T_u(R)/2$.
\end{proof}

\begin{lemma} \label{l-heights}
Let $G=(V,E)$ be a connected graph with a given vertex $u$ and let $\alpha$ and $\beta$ be two $3$-colourings of $G$.
For any $(\alpha\!\rightarrow\!\!\beta)$-recolouring $R=c_0, \ldots, c_\ell$, for each vertex $v$ in $V$, and for each $i$,
\begin{equation} \label{eq-heights}
2(\alpha(v)-c_i(v)) \equiv H_u^R(c_i,v) \bmod 6.
\end{equation}
\end{lemma}

\begin{proof}
We use induction on $i$.  If $i=0$, the LHS of (\ref{eq-heights}) is zero since $c_0=\alpha$ and the RHS of (\ref{eq-heights}) is zero by Lemma~\ref{l-absolute}.

Assume (\ref{eq-heights}) holds for $c_{i-1}$ by induction. Then in order to show (\ref{eq-heights}) holds for $c_i$, it is sufficient to show
\[
2(c_{i-1}(v) - c_{i}(v)) \equiv H^R_u(c_i,v) - H^R_u(c_{i-1},v)\bmod 6.
\]
This holds by Lemma~\ref{l-absolute}.
\end{proof}

Let $G$ be a connected graph with two $3$-colourings $\alpha$ and $\beta$. We define a function which gives a lower bound on the length of any $(\alpha\!\rightarrow\!\!\beta)$-recolouring. 

A \emph{focal} vertex $u^*=u^*(\beta)$ of $\beta$ is defined as follows. If $G$ has any fixed vertex with respect to $\beta$ then take $u^*$ to be an arbitrary fixed vertex. Otherwise, pick any vertex $u$ and order the vertices of $G$ according to their heights $h_{\alpha,u}(\beta, v)$: the order is the same irrespective of the choice of $u$ by Lemma~\ref{l-basechange}. Choose $u^*$ to be a median vertex in this order.
 
Now define 
\[ J(k) = J_{\alpha,\beta}(k) = \sum_{v \in V} |k+h_{\alpha,u^\ast}(\beta,v)|.
\]
Comparing with the definition of $T_{u^\ast}(R)$, we observe that if $k$ is the height of $u^\ast$ when $\beta$ is reached by a $(\alpha\!\rightarrow\!\!\beta)$-recolouring $R$, (i.e. $k = H^R_{u^\ast}(\beta, u^\ast)$), then $T_{u^\ast}(R)=J(k)$.  Thus we could find a lower bound on the length of the shortest $(\alpha\!\rightarrow\!\!\beta)$-recolouring if we could find the $k$ that minimises $J(k)$ for $k \equiv 2(\alpha(u^\ast)-\beta(u^\ast)) \bmod 6$ (this latter condition is required by Lemma~\ref{l-heights}). 

The next (non-graph-theoretic) proposition shows how to solve such optimisation problems in general.
\begin{proposition}
\label{pr:opt}
Suppose we have a sequence of integers $x_1, \ldots, x_n$. Let $x_{{\rm med}}$ be a median value in the sequence and assume $x_{{\rm med}}=0$. For each $k \in \mathbb{Z}$, let
$Q(k) := |k+x_1| + \cdots + |k+x_n|$.

Let $C \subseteq \mathbb{Z}$ and let $c^+ \geq 0$ (resp.\ $c^- \leq 0$) be the smallest non-negative (resp.\ largest non-positive) element in $C$. Then
\[
\min_{k \in C} Q(k) = \min \{ Q(c^+), Q(c^{-}) \}.
\] 
\end{proposition}
\begin{proof}
Suppose $k > c^+ \geq 0$ and observe that $k + x_{{\rm med}}>0$ is a median value in the sequence $s_k=k+x_1, \ldots, k+x_n$, and so this sequence contains at least as many positive values as non-positive values. Let $p$ (resp.\ $n$) be the number of positive (resp.\ non-positive) values in $s_k$. Then
\[
Q(k-1) = Q(k) -p +n \leq Q(k).
\]
Iterating this argument we see that $Q(c^+) \leq Q(k)$. An exactly analogous argument shows that if $k<c^-$ then $Q(c^-) \leq Q(k)$, proving the lemma.
\end{proof}

In order to apply this proposition, let us define $C_{\alpha,\beta}$ to be the set of integers congruent to $2(\alpha(u^\ast)-\beta(u^\ast))$ modulo $6$ and let $k_1,k_2$ be the smallest non-negative and largest non-positive integers in $C_{\alpha,\beta}$ (thus we have $(k_1,k_2) \in \{(0,0), (2,-4), (4,-2)\}$).

\begin{lemma}\label{l-j} 
Let $G$ be a connected graph, let $\alpha$ and $\beta$ be two $3$-colourings of $G$ and  let $u^\ast$ be a focal vertex of $\beta$. Suppose that $u^\ast$ is not fixed (with respect to $\beta$) and let $k \in C_{\alpha,\beta}$.
Then $J(k) \geq \min \{J(k_1), J(k_2)\}$.
\end{lemma}
\begin{proof}
Simply note that the sequence of numbers $h_{\alpha,u^\ast}(\beta,v)$ (which has median $h_{\alpha, u^\ast}(\beta,u^\ast)=0$) satisfies the premise of Proposition~\ref{pr:opt} with $Q=J$, $C=C_R$, $c^+=k_1$, and $c^-=k_2$.  
\end{proof}

\begin{lemma} \label{l-k-bound}
Let $G=(V,E)$ be a connected graph, let $\alpha$ and $\beta$ be two $3$-colourings of $G$ and  let $u^\ast$ be a focal vertex of $\beta$.
 For any $(\alpha\!\rightarrow\!\!\beta)$-recolouring $R$ of length $\ell$, $\ell \geq \frac{1}{2} \min \{J(k_1), J(k_2)\}$.
\end{lemma}
\begin{proof}
Let $k=H_{u^\ast}^R(\beta, u^\ast)$.  Using Lemma~\ref{l-totalheights} and the definitions, we have 
\begin{eqnarray*}
\ell & \geq & \frac{1}{2} T_{u^\ast}(R) \\
&=& \frac{1}{2} \left( \sum_{v \in V} |H_{u^\ast}^R(\beta,u^\ast) + h_{\alpha,u^\ast}(\beta,v)| \right) \\
& = & \frac{1}{2} \left( \sum_{v \in V} |k +h_{\alpha,u^\ast}(\beta,v)| \right) \\
&=& \frac{1}{2} J(k).
\end{eqnarray*} 
If $u^\ast$ is not fixed, the lemma follows from Lemma~\ref{l-j}.  If $u^\ast$ is fixed, then $k=H_{u^\ast}^R(\beta, u^\ast) = 0$ (by Lemma~\ref{l-absolute} and recalling that the colour of $u^\ast$ can never change). Thus $J(k)=J(0)=J(k_1)=J(k_2)$ and we are done.
\end{proof}

\begin{lemma} \label{l-k-time}
Let $G$ be a connected  graph, let $\alpha$ and $\beta$ be two $3$-colourings of $G$ and  let $u^\ast$ be a focal vertex of $\beta$.
The value of $\frac{1}{2}\min \{J(k_1), J(k_2)\}$ can be computed in time~$O(n)$.
\end{lemma}

\begin{proof}
All that is needed is to find the relative heights $h_{\alpha,u^\ast}(\beta,v)$ for each vertex $v$, which, as we noted in the proof of Lemma~\ref{l-heights-time}, can be done in time $O(n)$.
\end{proof}

\subsection{A Recolouring Algorithm}

In this subsection, we present an algorithm to find an $(\alpha\!\rightarrow\!\!\beta)$-recolouring $R$. 

First we show that we can focus on heights rather than colours.  That is, if we find a colouring that achieves certain values for the vertex heights, we will know the colours of the vertices.

\begin{lemma}\label{l-h}
Let $G=(V,E)$ be a connected graph with given vertex $u$ and let $\alpha$ and $\beta$ be two $3$-colourings of $G$.
Let $k \equiv 2(\alpha(u)-\beta(u)) \bmod 6$ be an integer.
If $R=c_0, \ldots, c_\ell$ is a recolouring such that for all $v \in V$, $H_u^R(c_\ell,v)=k+h_{\alpha,u}(\beta,v)$, then $c_\ell=\beta$.
\end{lemma}
\begin{proof}
All congruences are $\bmod \, 6$. We prove the lemma by induction on the number of vertices. Let $T$ be the spanning tree used to define $h$, and let $\{v_1, \ldots, v_n\}$ be a breadth-first ordering of $T$ with $v_1:=u$.

First we consider $u$.  As $h_{\alpha,u}(\beta, u)=0$ (by definition), we have $H_u^R(c_\ell, u)=k$ and so $H_u^R(c_\ell, u) \equiv 2(\alpha(u) - \beta(u))$.
Using Lemma~\ref{l-heights}, we find
\[
2(\alpha(u)-c_\ell(u)) \equiv 2(\alpha(u) - \beta(u)).
\]
which implies $c_\ell(u)=\beta(u)$.

Now assume that $c_\ell(v)=\alpha(v)$ for all $v \in \{v_1, \ldots, v_{k-1} \}$ and let $v_{k^*}$ be the ancestor of $v_k$ in $T$ (so $k^*<k$).
 By the premise of the lemma, we have $H_u^R(c_\ell,v_k)=k+h_{\alpha,u}(\beta,v_k)$. Also, 
\[
H_u^R(c_\ell,v_k) = H_u^R(c_\ell,u) + h_{\alpha,u}(c_\ell,v_k) = k + h_{\alpha,u}(c_\ell,v_k)
\]
where the first equality is by definition of $H$ and the second was noted above. Combining, we have $h_{\alpha,u}(\beta,v_k) = h_{\alpha,u}(c_\ell,v_k)$. Noting that 
\[
h_{\alpha,u}(\cdot,v_k) = h_{\alpha,u}(\cdot,v_{k^*}) + w(\cdot,\overrightarrow{v_{k^*}v_k}) - w(\alpha,\overrightarrow{v_{k^*}v_k}),
\] 
we deduce
\[
h_{\alpha,u}(\beta,v_{k^*}) + w(\beta,\overrightarrow{v_{k^*}v_k})
=
h_{\alpha,u}(c_\ell,v_{k^*}) + w(c_\ell,\overrightarrow{v_{k^*}v_k}).
\]
But since $\beta$ and $c_\ell$ are identical on $v_1, \ldots, v_{k^*}$, then $h_{\alpha,u}(\beta,v_{k^*}) =
h_{\alpha,u}(c_\ell,v_{k^*})$ which together with the above implies that  $w(\beta,\overrightarrow{v_{k^*}v_k}) = w(c_\ell,\overrightarrow{v_{k^*}v_k})$. Since $\beta(v_{k^*}) = c_\ell(v_{k^*})$, we have $\beta(v_{k}) = c_\ell(v_{k})$ as required.    
\end{proof}

\begin{lemma} \label{l-algo}
Let $G=(V,E)$ be a connected graph, let $\alpha$ and $\beta$ be two $3$-colourings of $G$ and  let $u^\ast$ be a focal vertex of $\beta$.
If $u^\ast$ is fixed with respect to $\beta$, let $k=0$;  otherwise, 
let $k \equiv 2(\beta(u^\ast)-\alpha(u^\ast)) \bmod 6$ be an integer.
If 
\begin{itemize}
\item[\emph{(A1)}] $F_{G,\alpha}^i=F_{G,\beta}^i$ for $i=1,2,3$,
\item[\emph{(A2)}] for each $vw \in E$,
$h_{\alpha,u^\ast}(\beta, v)-h_{\alpha,u^\ast}(\beta, w)+w(\beta,\overrightarrow{vw}) = w(\alpha,\overrightarrow{vw})$. 
\end{itemize}
then there exists an $(\alpha\!\rightarrow\!\!\beta)$-recolouring $R$ of length $\ell$ such that $\ell=\frac{1}{2} J(k)$.
\end{lemma}

\begin{proof}
We will define $R$ by describing how to recolour from $\alpha$ to a colouring $c$ such that, for all $v$, $H_{u^\ast}^R(c,v)=k+h_{\alpha,u^\ast}(\beta,v)$. Then, by Lemma~\ref{l-h}, $c=\beta$ as required.   Let $t(v)$  denote $k+h_{\alpha,u^\ast}(\beta,v)$.  This is the \emph{target} height of $v$.  When every vertex has reached its target, we are done.  Note that $J(k) = \sum_{v \in V} |t(v)|$.

In order to construct an $(\alpha\!\rightarrow\!\!\beta)$-recolouring of length $\ell = \frac{1}{2}J(k)$, it is sufficient to ensure that at each given stage, if $v$ is the vertex that changes colour, then this change reduces the difference between the current absolute height of $v$ and $t(v)$ by 2 (recall that by Lemma~\ref{l-absolute} the absolute height of $v$ changes by $2$ while all other absolute heights remain the same).

More definitions: for a vertex $u$ in~$G$  and colouring $c$, a \emph{rising path} from $u$ is a path on vertices $u=v_0, v_1,  \ldots v_t$ such that, for $1 \leq i \leq t$, $c(v_i) \equiv c(v_{i-1}) + 1 \bmod 3$. If $v_t$ has no neighbours coloured $c(v_t)+1 \bmod 3$ then the path is \emph{maximal} (and in this case we can recolour $v_t$ to $c(v_t)+1$ if we wish).
  A \emph{falling path} from $u$ is the same except that the colours decrease rather than increase moving along the path from $u$.  (That is, the colours along a rising path are, for example, $231231231231 \cdots$, and along a falling path are, for example, $321321321321 \cdots$)

It might not always be possible to find a maximal rising (or falling) path from a vertex $u$, but the following claim will be enough for us.

\begin{claim}
Let $R$ be an $(\alpha\!\rightarrow\!\!c)$-recolouring of $G$ and let
$v$ be a vertex of $G$.
If $t(v) - H_{u^\ast}^R(c,v) < 0$, then there is a maximal rising path $P$ from $v$, and, for every vertex $w$ on $P$,
\begin{equation} \label{c-rising}
t(w) - H_{u^\ast}^R(c,w) \leq  t(v)-H_{u^\ast}^R(c,v)<0.
\end{equation} 
If $t(v) - H_{u^\ast}^R(c,v) > 0$, then there is maximal falling path $Q$ from $v$, and, for every vertex $w$ on $Q$,
\begin{equation*} 
t(w) - H_{u^\ast}^R(c,w) \geq  t(v)-H_{u^\ast}^R(c,v)>0.
\end{equation*} 
\end{claim}

We will prove the first statement of the claim (the second can be proved in a similar way) by finding a maximal rising path.  To do this, we start with the trivial rising path on the single vertex $v$, and show that we can always extend the path or that it is maximal.  So suppose that we have found a rising path $w_0w_1 \cdots w_q$ where $w_0=v$ and each $w_i$, $0 \leq i \leq q$, satisfies (\ref{c-rising}).  

If $w_q$ has no neighbour coloured $c(w_q)+1$, the path is maximal and we are done.  

Suppose that $w_r$, $r<q$, is a neighbour of $w_q$ and is coloured $c(w_q)+1$.  Then $w_r, w_{r+1}, \ldots, w_q$ are fixed vertices (since the graph they induce is a fixed cycle coloured $\cdots 123123 \cdots$).
  But the absolute height of a fixed vertex is always 0 so $t(w_q)=H_{u^\ast}^R(c,w_q)=0$ contradicting~(\ref{c-rising}).  

The remaining possibility is that $w_q$ has a neighbour $w_{q+1}$ coloured $c(w_q)+1$ that is not already part of the rising path.   If we can show that (\ref{c-rising}) is satisfied with $w=w_{q+1}$, then we can extend the rising path to include $w_{q+1}$.  We have, using definitions and Lemma~\ref{l-heights-nec}, 
\begin{eqnarray*}
H_{u^\ast}^R(c,w_{q+1}) & = & H_{u^\ast}^R(c, u^\ast) + h_{\alpha, u^\ast}(c,w_{q+1}) \\
&=& H_{u^\ast}^R(c, u^\ast) + h_{\alpha, u^\ast}(c, w_q)+w(c,\overrightarrow{w_qw_{q+1}}) - w(\alpha,\overrightarrow{w_qw_{q+1}}) \\
&=& H_{u^\ast}^R(c, w_q) +w(c,\overrightarrow{w_qw_{q+1}}) - w(\alpha,\overrightarrow{w_qw_{q+1}}).
\end{eqnarray*} 
And using (A2), we find
\begin{eqnarray*}
t(w_{q+1}) & = & k + h_{\alpha, u^\ast}(\beta,w_{q+1}) \\
& = & k + h_{\alpha, u^\ast}(\beta,w_q) +w(\beta,\overrightarrow{w_qw_{q+1}}) - w(\alpha,\overrightarrow{w_qw_{q+1}}) \\
& = & t(w_q) +w(\beta,\overrightarrow{w_qw_{q+1}}) - w(\alpha,\overrightarrow{w_qw_{q+1}}).
\end{eqnarray*}
Subtracting 
\begin{equation*}
t(w_{q+1}) - H_{u^\ast}^R(c,w_{q+1}) = t(w_{q}) - H_{u^\ast}^R(c,w_{q}) + w(\beta,\overrightarrow{w_qw_{q+1}}) - w(c,\overrightarrow{w_qw_{q+1}})
\end{equation*}
Noting that $w(c,\overrightarrow{w_qw_{q+1}}) = 1 \geq  w(\beta,\overrightarrow{w_qw_{q+1}})$, we have
\begin{equation*}
t(w_{q+1}) - H_{u^\ast}^R(c,w_{q+1}) \leq t(w_{q}) - H_{u^\ast}^R(c,w_{q}),
\end{equation*}
and the claim is proved.

\medskip
We now inductively describe how to obtain an $(\alpha\!\rightarrow\!\!\beta)$-recolouring of $G$ of length $\ell = \frac{1}{2}J(k)$.
Suppose we have a partial recolouring $R$ from $\alpha$ to $c$.
Recall that it is sufficient for us to specify which vertex should be recoloured (to give a new proper colouring) and to show that this change reduces the difference between the current height and the target height of $v$ (while keeping all other heights unchanged).

\begin{enumerate}
\item Find a vertex $x$ for which $|t(x)- H_{u^\ast}^R(c,x)|$ is maximum.
\item If $t(x)-H_{u^\ast}^R(c,x)>0$, find a maximal rising path from $x$.  Else find a maximal falling path from $x$.  In either case, let $v$ be the end-vertex of the path.  
\item Change the colour of $v$ so that $|t(v)-H_{u^\ast}^R(c,v)|$ is reduced by 2.
\end{enumerate}
Consider the case where $t(x)-H_{u^\ast}^R(c,x)>0$ (the other case is analagous).  Then, by the Claim,  $t(v)-H_{u^\ast}^R(c,v)>0$.  By Lemma~\ref{l-absolute}, we can increase $H_{u^\ast}^R(c,v)$ by $2$ (while keeping the absolute heights of all other vertices unchanged) by increasing the colour of $v$ and this reduces $|t(x)-H_{u^\ast}^R(c,x)|$ by $2$. As $v$ is at the end of a maximal rising path this increase in colour results in a proper colouring, as required.  
\end{proof}

\medskip
\noindent
{\it Proof of Theorem~\ref{t-ptime}}.
Let $G=(V,E)$ be a graph and let $\alpha$ and $\beta$ be two $3$-colourings of $G$. Assume $G$ is connected; otherwise consider each component separately.
By Lemmas~\ref{l-Gfc} and~\ref{l-heights-nec}, a path between $\alpha$ and $\beta$ in $R_3(G)$ only exists if 
\begin{itemize}
\item $F_{G,\alpha}^i=F_{G,\beta}^i$ for $i=1,2,3$,
\item for each $vw \in E$,
$h_u(\beta, v)-h_u(\beta, w)+w(\beta,\overrightarrow{vw}) = w(\alpha,\overrightarrow{vw})$. 
\end{itemize}
and by Lemma~\ref{l-algo} these conditions are also sufficient.  By Lemmas~\ref{l-fixed-time} and~\ref{l-heights-time}, these conditions can be tested in time $O(n+m)$.  Moreover if a path between $\alpha$ and $\beta$ does exist, then, by Lemma~\ref{l-k-bound}, it has length at least $\frac{1}{2} \min \{J(k_1), J(k_2) \}$, and, by Lemma~\ref{l-algo}, a path of exactly this length does exist.  By Lemma~\ref{l-k-time}, the value of $\frac{1}{2} \min \{J(k_1), J(k_2) \}$ can be found in time $O(n+m)$ 
and so the length of a shortest path between $\alpha$ and $\beta$ can be found in time $O(n+m)$.
This implies Theorem~\ref{t-ptime}.
\qed

\bigskip

It is straightforward to see that $\frac{1}{2} \min \{J(k_1), J(k_2) \}$ is $O(n^2)$, and
we note that in~\cite{CHJ06b}, examples of families of graphs with pairs of 3-colourings at distance $\Omega(n^2)$ were given.  The purpose of the description of the algorithm in Lemma~\ref{l-algo} was to estabish the sufficiency of the necessary conditions of Lemmas~\ref{l-Gfc} and~\ref{l-heights-nec} and its running time is not optimised, it can easily be adapted to run in time $O(n^2)$.  The key is to note that once a vertex $x$ for which $|t(x)- H^R(c,x)|$ is maximum is found, then $|t(v)- H^R(c,v)|$ is maximum for every vertex on a maximal rising path from $x$ (if $t(x)- H^R(c,x)>0$, else consider falling paths) and each of these vertices must be recoloured in turn.

\section{An FPT Algorithm for $k$-Colouring Reconfiguration}\label{s-fpt}

In this section we will present our \FPT\ algorithm for \textsc{$k$-Colouring Reconfiguration} when parameterized by~$\ell$.   
Let $G=(V,E)$ be a graph on $n$ vertices, and let $\alpha$, $\beta$ be two proper $k$-colourings of $G$.
First we prove 
three 
lemmas concerning the vertices that might be recoloured if a path between $\alpha$ and $\beta$ of length at most $\ell$ does exist.
That is, we assume that $(G,\alpha,\beta,\ell)$ is a yes-instance of {\sc $k$-Colouring Reconfiguration}.
This means that  there exists an $(\alpha\!\rightarrow\!\!\beta)$-recolouring $R= c_0, \dots,c_{\ell}$. We assume that $R$ has {\em minimum length}.  

We say that~$R$ \emph{recolours} a vertex $u$ if $c_q(u) \neq \alpha(u)$ for some $q$.   Notice that if for each recoloured vertex $u$ we find the least $q$ such that $c_q(u) \neq \alpha(u)$, these values must be distinct (else $c_q$ and $c_{q-1}$ disagree on more than one vertex).  Thus the number of distinct vertices recoloured by $R$ is at most $\ell$. We will prove something stronger. For $0 \leq q \leq \ell$, let $W_q$ be the set of vertices on which $c_0$ and~$c_q$ disagree,
that is, $W_q = \{u\in V : c_0(u)\neq c_q(u)\}$.
\begin{lemma} \label{lemma-wq}
For all $q$ with $1\leq q\leq \ell$, the set $W_q$ has size $|W_q| \leq q$. 
\end{lemma}
\begin{proof}
Suppose this is false and let $r$ be the smallest value such that $|W_r|>r$.   So $|W_{r-1}| \leq r-1$ (clearly $r-1\geq0$ as $W_0$ is the empty set). 
Then there are (at least) two vertices $v_1, v_2$ in $W_{r} \setminus W_{r-1}$, and so, for $i \in \{1,2\}$, $c_{r-1}(v_i) = c_0(v_i) \neq c_{r}(v_i)$, and $c_r$ and $c_{r-1}$ disagree on more than one vertex;
a contradiction.
\end{proof}
For any $u \in V$, let $N(u)$ be the set of neighbours of $u$.  For any  $v \in N(u)$,  let  $N(u,v)=\{w \in N(u) : \alpha(w)=\alpha(v) \}$; that is, the set of neighbours of $u$ with the same colour as $v$ in $\alpha$.
Let $A_0 = \{v \in V : \alpha(v) \neq \beta(v) \}$ be the set of vertices on which $\alpha$ and $\beta$ disagree.  For $i \geq 1$, let 
$A_i  =  \bigcup_{u \in A_{i-1}}  \{v \in N(u) : |N(u,v)| \leq \ell \}$.
That is, to find $A_{i}$ consider each vertex $u$ in $A_{i-1}$ and partition $N(u)$ into colour classes (according to the colouring $\alpha$).  Vertices in $N(u)$ that belong to colour classes of size at most $\ell$ belong to $A_{i}$.
Note that two sets $A_h$ and $A_i$ need not be disjoint.

Our first goal is to show that each vertex recoloured by $R$ must be in $A^*=\bigcup_{h=0}^{\ell-1} A_h$. 
We will then show that the size of $A^*$ is bounded by a function of $k+\ell$. 
This will enable us to use brute-force to find $R$ or some other 
 $(\alpha\!\rightarrow\!\!\beta)$-recolouring of~$G$ (if it exists).

\begin{lemma} \label{lemma-R}
Each vertex recoloured by $R$ belongs to $A^*$. 
\end{lemma}

\begin{proof}
For $i\geq 0$, let $L_i=A_i \setminus (\bigcup_{h<i} A_h)$ be the set of vertices that are in $A_i$ but not in any $A_h$ with $h<i$.
Let $z$ be the greatest value such that $R$ recolours a vertex in $L_z$; denote this vertex by $v_z$.
By definition, every vertex in~$A_0$ is recoloured by $R$. Let $v_0\in A_0$.
We claim that also for $1 \leq i \leq z-1$, there is a vertex $v_i \in L_i$ that is recoloured by~$R$. 
Then, as  $v_0,\ldots,v_z$ are distinct vertices and $R$ has length~$\ell$, we have $z\leq \ell-1$ proving the lemma.
For contradiction, assume there is a set $L_i$ ($1\leq i\leq z-1$) that contains no vertex recoloured by~$R$.

From $R$ we construct a new recolouring sequence $R'$ by ignoring
every recolouring step done to a vertex in $V \setminus \bigcup_{h<i} L_h$.  
For $0 \leq q \leq \ell$, let $d_q$ be a colouring of $G$ such that
\begin{itemize}
\item if $u \in \bigcup_{h<i} L_h$, $d_q(u)=c_q(u)$;
\item if $u \notin \bigcup_{h<i} L_h$, $d_q(u) = \alpha(u)$.
\end{itemize}
Let $R'$ be the sequence $d_0, \dots, d_{\ell}$.  Note that $d_0=\alpha$, 
as $d_0(u)$ is either $c_0(u)$ or $\alpha(u)$, and $c_0 = \alpha$.  Moreover, if $u \in \bigcup_{h<i} L_h=\bigcup_{h<i}A_i$ then
$d_{\ell}(u)=c_{\ell}(u)=\beta(u)$, and if $u \notin \bigcup_{h<i} L_h$ then $d_{\ell}(u)=\alpha(u)=\beta(u)$ (since $\alpha$ and $\beta$ only disagree on vertices in $A_0$); thus $d_{\ell}=\beta$.
This means that if we can show that $d_1,\ldots,d_{\ell-1}$ are proper colourings, then $R'$ is an $(\alpha\!\rightarrow\!\!\beta)$-recolouring. We will prove this first.

Assume to the contrary that $R'$ contains a colouring $d_q$ that is not proper. Then there is an edge~$uv$ with $d_q(u) = d_q(v)$.  If $u$ and $v$ both belong to $\bigcup_{h<i} L_h$ then $c_q(u) = c_q(v)$, and if neither belong to $\bigcup_{h<i} L_h$ then $\alpha(u) = \alpha(v)$. Both cases are not possible, as $c_q$ and~$\alpha$ are proper colourings.
Hence we may assume, without loss of generality, that $u \in \bigcup_{h<i} L_h$ and $v \notin \bigcup_{h<i} L_h$.
Then $c_q(u)=d_q(u)=d_q(v)=\alpha(v)$ by the definition of $d_q$. 

As $v \in N(u)$, the set $N(u,v)$ exists. First suppose $|N(u,v)| \leq \ell$. Then $v\in A_{i}$ by the definition of $A_i$. Hence $v\in L_h$ for some $h\le i$. 
As $v \notin \bigcup_{h<i} L_h$,  we obtain $v\in L_i$. 
By assumption, no vertex of $L_i$ is recoloured by $R$. Hence $c_q(v)=\alpha(v)$
and thus $c_q(u)=c_q(v)$ contradicting the fact that $c_q$ is a proper $k$-colouring.   

Now suppose $|N(u,v)|>\ell$. 
Because $c_q(u)=\alpha(v)$ and $c_q$ is proper, we find that $c_q(w) \neq c_q(u)=\alpha(v)=\alpha(w)$ for all $w \in N(u,v)$.
Thus $W_q \supseteq N(u,v)$ and so $|W_q| \geq |N(u,v)| > \ell\geq q$ contradicting the fact that $|W(q)|\leq q$ by Lemma~\ref{lemma-wq}.
So, $d_q$ must be proper. We conclude that $R'$ is an $(\alpha\!\rightarrow\!\!\beta)$-recolouring of length~$\ell$.

We now proceed as follows. 
Recall that $v_z\in L_z$.
Then there is a pair of colourings $c_q$ and $c_{q+1}$ that differ only on $v_z$.  
Because $v_z\in L_z$,  $v_z \notin \bigcup_{h<i} L_h$. Hence,
 $d_q$ and $d_{q+1}$ are identical colourings. We remove $d_q$ from~$R'$ to obtain another $(\alpha\!\rightarrow\!\!\beta)$-recolouring, which has length shorter than $\ell$, contradicting that~$R$ has minimum length. This completes the proof.
\end{proof}
We now show that 
 $|A^*|$ is bounded by a function that depends only on~$k$ and $\ell$.

\begin{lemma}\label{lemma-sp}
The set $A^*$ has size $|A^*| \le \ell \cdot (k \ell)^{\ell}$.  
\end{lemma}
\begin{proof}
Recall that $A_0$ is the set of vertices on which $\alpha$ and $\beta$ disagree. 
Hence,
all vertices of~$A_0$ need to be recoloured by $R$.
Thus we have  $|A_0|\le \ell$. 
Now let $i$ be such that $1\le i \le \ell -1$.
Recall that the set $A_i$ is defined as $A_i =  \bigcup_{u \in A_{i-1}}  \{v \in N(u) : |N(u,v)| \leq \ell \}$.
By this recursive definition, each vertex of $A_i$ is a neighbour 
of a vertex of $A_{i-1}$, and each vertex of $A_{i-1}$ has at most $k\cdot \ell$ 
neighbours in $A_i$. Consequently $|A_i| \le |A_{i-1}| \cdot k\cdot \ell$, for all
$1\le i \le \ell -1$.
Hence,
$ |A^*| = \sum_{i=0}^{\ell -1} |A_i| \le  \sum_{i=0}^{\ell -1}  \ell \cdot (k\cdot \ell)^{i}
\le \ell \cdot \frac{(k\cdot \ell)^{\ell} -1}{k\cdot \ell -1}\le \ell \cdot (k\cdot \ell)^{\ell}$.
\end{proof} 

\medskip
\noindent
We are now ready to present our \FPT\ algorithm and prove Theorem~\ref{t-fpt}.

\medskip
\noindent
{\it Proof of Theorem~\ref{t-fpt}.}
Let $k\geq 1$, and let $(G,\alpha,\beta,\ell)$ be an instance of {\sc $k$-Colouring Reconfiguration}, where $G$ is a graph on $n$ vertices, and
$\alpha,\beta$ are two proper $k$-colourings of $G$.
Our algorithm does as follows. First compute the set $A^*$ in $O(n^2)$ time.
By Lemma~\ref{lemma-sp}, we find that $|A^*|\leq  \ell \cdot (k \ell)^{\ell}$.
By Lemma~\ref{lemma-R}, we only have to search for a  path of length at 
most~$\ell$ in $R_k(G)$ among the vertices of $A^*$.  
By allowing consecutive recolourings to be equal we may restrict our search to
$(\alpha\!\rightarrow\!\!\beta)$-recolourings of length exactly $\ell$. 
Use brute force to enumerate all possible sequences of pairs $(v_i,c_i)$,
such that for all $0\le i \le \ell-1$, $v_i$ is a vertex in~$A^*$ and $c_i$ is a 
colour in $\{1,\ldots,k\}$. 
For each such sequence do as follows.
Starting from~$\alpha$, recolour  $v_i$ with colour $c_i$ for $i=0,\ldots, \ell-1$.  As soon as this results in a $k$-colouring that is not proper, 
stop considering the sequence. If not, check whether the resulting colouring is equal to $\beta$. If this happens, then there is a path of length $\ell$
in $R_k(G)$.  Hence, return {\tt yes}. Otherwise, that is, if no sequence has this property, return {\tt no}.  Processing one sequence takes time $O(\ell n^2)$. By using Lemma~\ref{lemma-sp}, the number of sequences is at most $(|A^*| \cdot k)^{\ell} \le  ((\ell \cdot (k\cdot \ell)^{\ell})\cdot k)^{\ell} \le (k\cdot \ell)^{\ell^2 +\ell}$, leading to a total running time of  $O((k\cdot \ell)^{\ell^2 +\ell} \cdot \ell n^2)$.  This completes the proof.\qed

\section{A Lower Bound for Kernelization for $k\geq 4$}\label{s-nokernel}

This section provides a proof for Theorem~\ref{t-ker}, i.e., that $k$-\textsc{Colouring Reconfiguration}, for~$k\geq 4$, does not admit a polynomial kernelization in terms of~$\ell$, unless $\mathsf{NP\subseteq coNP/poly}$ (it is known that the latter would imply a collapse of the polynomial hierarchy).  To prove the result we give a so-called polynomial parameter transformation from a problem that, assuming $\mathsf{NP\nsubseteq coNP/poly}$, is known not to admit a polynomial kernelization and also no polynomial compression.\footnote{A (polynomial) compression is a relaxed form of (polynomial) kernelization: The output may be with respect to \emph{any (possibly unparameterized) problem}.} A polynomial parameter transformation, short PPT, is a standard Karp reduction with the additional property that the parameter value of the returned instance is polynomially bounded in the parameter of the input instance. It is well known and easy to see that a PPT from a source problem without polynomial compression implies that the target problem admits no polynomial compression and hence also no polynomial kernelization (cf.~\cite{BJK14}).

As our source problem we use \textsc{Hitting Set}. This problem takes as input a finite set $U$, a set ${\cal F}\subseteq 2^U$ and an integer~$p$, and asks whether there exists a {\it hitting set} $S\subseteq U$ of size at most~$p$, that is, a set $S$ with $|S|\leq p$ such that every $F\in {\cal F}$ contains at least one element of $S$.  The {\sc Hitting Set} problem can also be formulated as the \textsc{Red-Blue Dominating Set} problem, which takes a bipartite graph with partition classes~$R$ and $B$ and an integer $k$, and asks whether there exists a set $D\subseteq R$ of size at most $k$ such that every vertex of $B$ has at least one neighbour in $D$.  Dom, Lokshtanov, and Saurabh~\cite{DLS09} showed that the \textsc{Red-Blue Dominating Set} problem, parameterized by $k+|B|$, does not admit a polynomial kernelization (unless $\mathsf{\NP\subseteq coNP/poly}$). As $k\leq |B|$ holds for any non-trivial instance, the same result holds with parameter $|B|$ instead of $k+|B|$.   Since the result for \textsc{Red-Blue Dominating Set} makes use of the standard framework of giving or/and-compositions, it is known to also rule out polynomial compressions (cf.~\cite{BJK14}).

\begin{lemma}[\cite{DLS09}]\label{l-hitting}
The \textsc{Hitting Set} problem parameterized by $|{\cal F}|$ does not admit a polynomial compression unless  $\mathsf{\NP\subseteq coNP/poly}$.
\end{lemma}
We are now ready to prove Theorem~\ref{t-ker}.  The main idea for the reduction is to create a~$4$-coloured tree that serves as a selection gadget for each set, which requires a recolouring at its root. This in turn requires a chain of earlier recolourings starting in one of the leaves; the selection of possible leaves encodes the elements of the set. Finally, recolouring any leaf requires a recolouring in a set of vertices corresponding to the ground set; this encodes the selection of a hitting set. Crucially, the height of the tree construction, which factors into the number~$\ell$ of needed recolourings, can be bounded polynomially in the input parameter~$m=|\F|$.

\medskip
\noindent
{\it Proof of Theorem~\ref{t-ker}.}  By Lemma~\ref{l-hitting} it suffices to show that there is a polynomial parameter transformation from \textsc{Hitting Set} parameterized by $|{\cal F}|$  to $k$-\textsc{Colouring Reconfiguration} parameterized by $\ell$. 

Let $(U,\F,p)$ be an instance of \textsc{Hitting Set}. Let $m=|\F|$. 
We give a polynomial-time construction of an equivalent instance~$(G,\alpha,\beta,\ell)$ of the $k$-\textsc{Colouring Reconfiguration} problem, 
where~$\ell$ is polynomially bounded in~$m$. Note that~$p<m$ or else the instance~$(U,\F,p,m)$ is trivially yes and our transformation is trivial.

\medskip
\noindent
\emph{Construction.} We begin with a standard argument for bounding the size of~$U$: If any two elements of~$U$ occur exactly in the same sets of~$\F$ then we will never need both for a minimum hitting set and can safely discard either one of them. Thus, without loss of generality, we may assume that no two such elements exist, which implies that~$|U|\leq 2^{|\F|}=2^m$. Thus, our final parameter value~$\ell$ may depend polynomially on~$\log |U|=O(m)$. For convenience let~$n=|U|$ and let~$U=\{1,\ldots,n\}$.

The graph~$G$ will consist of four components:
\begin{enumerate}
 \item Two adjacent vertices~$s,t$ with~$\alpha(s)=\beta(t)=2$ and~$\alpha(t)=\beta(s)=3$. These are the only two vertices with different colours in~$\alpha$ and~$\beta$.
 \item A clique of~$k$ vertices~$u_1,\ldots,u_k$ with colours~$\alpha(u_i)=\beta(u_i)=i$ that will be used to control permissible colours for all other vertices.
 \item An independent set of vertices~$v_1,\ldots,v_n$, one for each element of~$U$, each with 
colour~$\alpha(v_i)=\beta(v_i)=4$ that will be used to simulate selection of a hitting set of size at most~$p$.
 \item One \emph{selection gadget} for each set~$F\in\F$ that simulates a selection of one element of the hitting set to hit~$F$. These will be described later as they are somewhat more involved.
\end{enumerate}
Clearly, since each vertex of the clique is adjacent to all other colours but its own, it is impossible to recolour any vertex~$u_i$ and obtain a proper~$k$-colouring. Thus, we can use adjacency to parts of the clique to forbid certain colours from being used for other vertices. Generally, all other vertices are made adjacent to all of~$u_5,\ldots,u_k$, effectively reducing the setting to the case that~$k=4$. (Mainly this ensures that our reduction works for all values~$k\geq 4$.)
Additionally, all vertices~$v_1,\ldots,v_n$ are made adjacent to~$u_2$ and~$u_3$, which in total allows only colours~$1$ and~$4$ to be used for the independent set.

Finally, we restrict vertex~$s$ to colours~$2$ and~$3$ and vertex~$t$ to colours~$2$,~$3$, and~$4$. Note that if only~$2$ and~$3$ were possible for both~$s$ and~$t$, then it would be impossible to recolour even just the graph on~$s$ and~$t$. Using the additional option of colour~$4$ for~$t$ the following sequence works: (1) recolour~$t$ to~$4$, (2) recolour~$s$ to~$3$, and (3) recolour~$t$ to~$2$. The hitting set question will be encoded in a part of the graph that requires recolouring in order not to obstruct colouring~$t$ with colour~$4$ (and will be reverted once~$t$ is coloured~$2$).

We will now describe the construction of the selection gadgets.
The basic building block is a claw on vertices~$a_\dag,b_\dag,c_\dag,d_\dag$ with~$c_\dag$ the center vertex (adjacent to~$a_\dag,b_\dag,d_\dag$) and with the following $\alpha$ and $\beta$ colours and forbidden colours:
\begin{enumerate}
 \item For~$a_\dag$ we have~$\alpha(a_\dag)=\beta(a_\dag)=2$, and, using adjacency to the~$k$-clique, only colours~$2$ and~$4$ allow proper~$k$-colourings.
 \item Similarly, for~$b_\dag$ we have~$\alpha(b_\dag)=\beta(b_\dag)=3$, and only colours~$3$ and~$4$ are possible.
 \item For the center vertex~$c_\dag$ we have~$\alpha(c_\dag)=\beta(c_\dag)=1$, and only colours~$1$,~$2$, and~$3$ are possible.
 \item For vertex~$d_\dag$ we have~$\alpha(d_\dag)=\beta(d_\dag)=4$, and only colours~$1$ and~$4$ are possible.
\end{enumerate}
The idea is that connecting such claws in a tree-like fashion gives the desired selection gadget. 
For the basic functionality that is to recolour~$d_\dag$ with~$1$ it is necessary to first recolour~$c_\dag$ to either~$2$ or~$3$. This in turn first requires a recolouring of (accordingly) either~$a_\dag$ to~$4$ or~$b_\dag$ to~$4$. Now if both~$a_\dag$ and~$b_\dag$ are adjacent to, say,~$d_\ddag$ and~$d_{\ddag'}$ of further such claws then the same argumentation continues since we again would need to recolour first~$d_\ddag$ from~$4$ to~$1$ or~$d_{\ddag'}$ from~$4$ to~$1$.
 
Now, let us describe the tree-like arrangement in more detail. For convenience, let us assume that~$n=2^r$ for some integer~$r$, which we can achieve by adding at most~$n-1$ dummy elements to~$U$ that never occur in any set (thereby at most doubling~$n$). We make copies of the claw construction that we just explained, for all values of
\[
\dag\in\{(F,x,y)\mid F\in\F, x\in\{0,\ldots,r-1\}, y\in\{1,\ldots,2^x\}\}.
\]
We connect these claws as follows (see Figure~\ref{figure:lowerboundtree}):
\begin{enumerate}
 \item Each vertex~$d_{F,0,1}$ is made adjacent to the vertex~$t$.
 \item Each vertex~$d_{F,x,y}$, with~$x\in\{1,\ldots,r-1\}$ is made adjacent to~$a_{F,x-1,(y+1)/2}$ if~$y$ is odd, and to~$b_{F,x-1,y/2}$ if~$y$ is even.
 \item Each vertex~$a_{F,r-1,y}$ is made adjacent to vertex~$v_{2y-1}$ of the independent set.
 \item Each vertex~$b_{F,r-1,y}$ is made adjacent to vertex~$v_{2y}$.
\end{enumerate}
The idea is that to recolour~$d_{F,0,1}$ to~$1$ it is ultimately necessary to first recolour some vertex~$v_j$ in the independent set from~$4$ to~$1$. This in turn allows to recolour the adjacent~$a_{F,\cdot,\cdot}$ or~$b_{F,\cdot,\cdot}$ vertex to~$4$ and then propagate possible recolourings towards~$d_{F,0,1}$.

\tikzstyle{vertex}=[rectangle split, rectangle split parts=2,draw]

\begin{figure}[p]
\centering
\begin{tikzpicture}[rotate=90,scale=0.7,every text node part/.style={text centered}]

\node[vertex] (s) at (0,0) {$\{\boldsymbol{2},\underline{3}\}$\nodepart{second}$s$};
\node[vertex] (t) at (0,-3) {$\{\underline{2},\boldsymbol{3},4\}$\nodepart{second}$t$};

\node[vertex] (dF01) at (-6,-3) {$\{1,\underline{\boldsymbol{4}}\}$\nodepart{second}$d_{F,0,1}$};
\node[vertex] (cF01) at (-6,-6) {$\{\underline{\boldsymbol{1}},2,3\}$\nodepart{second}$c_{F,0,1}$};
\node[vertex] (bF01) at (-3,-6) {$\{\underline{\boldsymbol{3}},4\}$\nodepart{second}$b_{F,0,1}$};
\node[vertex] (aF01) at (-9,-6) {$\{\underline{\boldsymbol{2}},4\}$\nodepart{second}$a_{F,0,1}$};

\node[vertex] (dF11) at (-9,-8.5) {$\{1,\underline{\boldsymbol{4}}\}$\nodepart{second}$d_{F,1,1}$};
\node[vertex] (cF11) at (-9,-11) {$\{\underline{\boldsymbol{1}},2,3\}$\nodepart{second}$c_{F,1,1}$};
\node[vertex] (bF11) at (-7,-11) {$\{\underline{\boldsymbol{3}},4\}$\nodepart{second}$b_{F,1,1}$};
\node[vertex] (aF11) at (-11,-11) {$\{\underline{\boldsymbol{2}},4\}$\nodepart{second}$a_{F,1,1}$};

\node[vertex] (dF12) at (-3,-8.5) {$\{1,\underline{\boldsymbol{4}}\}$\nodepart{second}$d_{F,1,2}$};
\node[vertex] (cF12) at (-3,-11) {$\{\underline{\boldsymbol{1}},2,3\}$\nodepart{second}$c_{F,1,2}$};
\node[vertex] (bF12) at (-1,-11) {$\{\underline{\boldsymbol{3}},4\}$\nodepart{second}$b_{F,1,2}$};
\node[vertex] (aF12) at (-5,-11) {$\{\underline{\boldsymbol{2}}\}$\nodepart{second}$a_{F,1,2}$};

\node[vertex] (dFF01) at (+6,-3) {$\{1,\underline{\boldsymbol{4}}\}$\nodepart{second}$d_{F',0,1}$};
\node[vertex] (cFF01) at (+6,-6) {$\{\underline{\boldsymbol{1}},2,3\}$\nodepart{second}$c_{F',0,1}$};
\node[vertex] (bFF01) at (+9,-6) {$\{\underline{\boldsymbol{3}},4\}$\nodepart{second}$b_{F',0,1}$};
\node[vertex] (aFF01) at (+3,-6) {$\{\underline{\boldsymbol{2}},4\}$\nodepart{second}$a_{F',0,1}$};

\node[vertex] (dFF11) at (+3,-8.5) {$\{1,\underline{\boldsymbol{4}}\}$\nodepart{second}$d_{F',1,1}$};
\node[vertex] (cFF11) at (+3,-11) {$\{\underline{\boldsymbol{1}},2,3\}$\nodepart{second}$c_{F',1,1}$};
\node[vertex] (bFF11) at (+5,-11) {$\{\underline{\boldsymbol{3}},4\}$\nodepart{second}$b_{F',1,1}$};
\node[vertex] (aFF11) at (+1,-11) {$\{\underline{\boldsymbol{2}}\}$\nodepart{second}$a_{F',1,1}$};

\node[vertex] (dFF12) at (+9,-8.5) {$\{1,\underline{\boldsymbol{4}}\}$\nodepart{second}$d_{F',1,2}$};
\node[vertex] (cFF12) at (+9,-11) {$\{\underline{\boldsymbol{1}},2,3\}$\nodepart{second}$c_{F',1,2}$};
\node[vertex] (bFF12) at (+11,-11) {$\{\underline{\boldsymbol{3}},4\}$\nodepart{second}$b_{F',1,2}$};
\node[vertex] (aFF12) at (+7,-11) {$\{\underline{\boldsymbol{2}},4\}$\nodepart{second}$a_{F',1,2}$};

\node[vertex] (v1) at (-9,-15) {$\{1,\underline{\boldsymbol{4}}\}$\nodepart{second}$v_1$};
\node[vertex] (v2) at (-3,-15) {$\{1,\underline{\boldsymbol{4}}\}$\nodepart{second}$v_2$};
\node[vertex] (v3) at (3,-15) {$\{1,\underline{\boldsymbol{4}}\}$\nodepart{second}$v_3$};
\node[vertex] (v4) at (9,-15) {$\{1,\underline{\boldsymbol{4}}\}$\nodepart{second}$v_4$};

\draw[very thick] (s) -- (t);

\draw[very thick] (t) -- (dF01) -- (cF01);
\draw[very thick] (aF01) -- (cF01) -- (bF01);

\draw[very thick] (aF01) -- (dF11) -- (cF11);
\draw[very thick] (aF11) -- (cF11) -- (bF11);

\draw[very thick] (bF01) -- (dF12) -- (cF12);
\draw[very thick] (aF12) -- (cF12) -- (bF12);

\draw[very thick] (t) -- (dFF01) -- (cFF01);
\draw[very thick] (aFF01) -- (cFF01) -- (bFF01);

\draw[very thick] (aFF01) -- (dFF11) -- (cFF11);
\draw[very thick] (aFF11) -- (cFF11) -- (bFF11);

\draw[very thick] (bFF01) -- (dFF12) -- (cFF12);
\draw[very thick] (aFF12) -- (cFF12) -- (bFF12);

\draw[very thick] (aF11.east) -- (v1.west);
\draw[very thick] (bF11.east) -- (v2.west);
\draw[very thick] (aF12.east) -- (v3.west);
\draw[very thick] (bF12.east) -- (v4.west);

\draw[very thick] (aFF11.east) -- (v1.west);
\draw[very thick] (bFF11.east) -- (v2.west);
\draw[very thick] (aFF12.east) -- (v3.west);
\draw[very thick] (bFF12.east) -- (v4.west);

\end{tikzpicture}
\caption{\label{figure:lowerboundtree} A small example of the lower bound construction for $(U,\F,p)$ and~$k=4$ when~$U=\{1,2,3,4\}$ and~$\F=\{F,F'\}$ with~$F=\{1,2,4\}$ and~$F'=\{2,3,4\}$. For ease of presentation the clique on vertices~$u_1,\ldots,u_4$ is not shown; instead sets in the nodes state the allowed colours, where the boldface number is the initial (alpha) colour and the underlined number is the target (beta) colour.}
\end{figure}
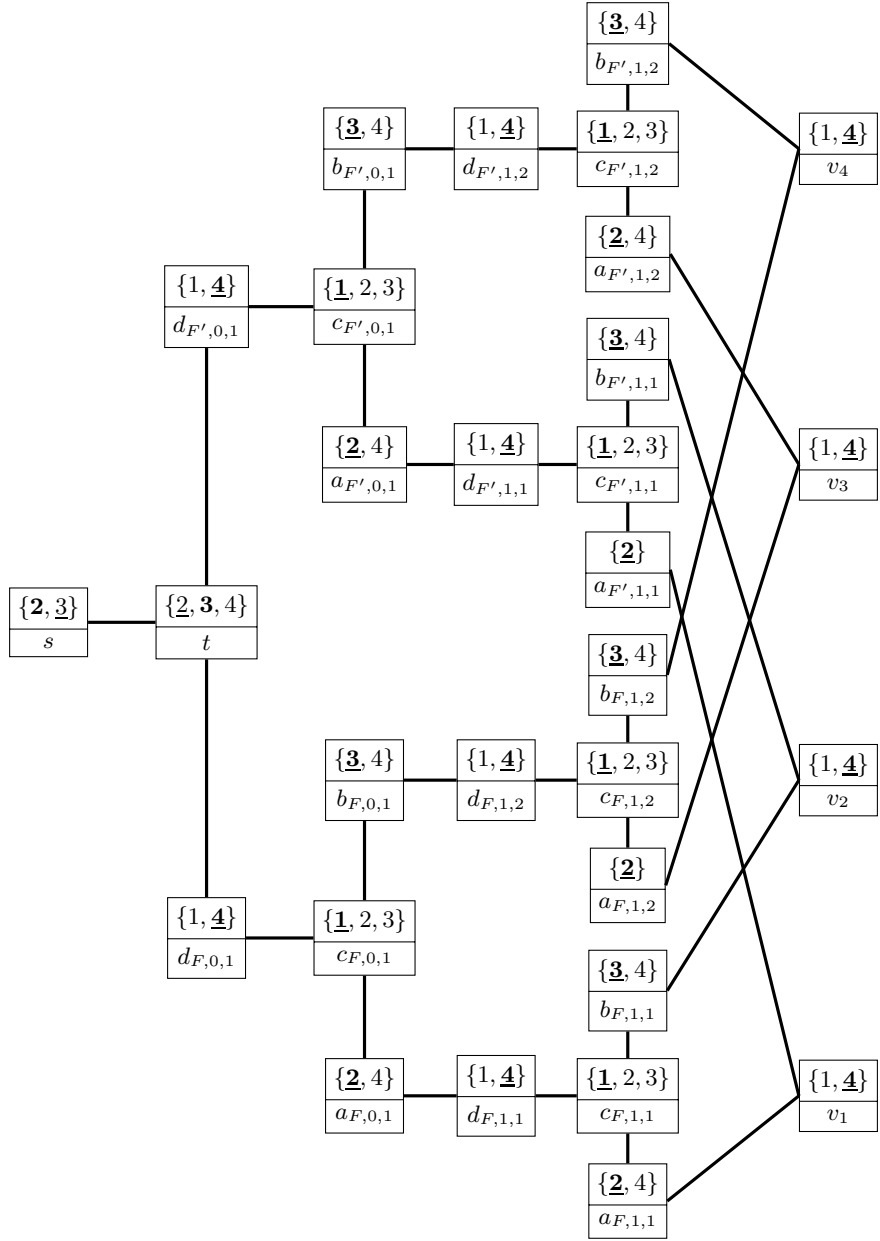

Note that so far we have made no distinction between trees made for different sets~$F\in\F$; we now make the following modifications to vertices~$a_{F,r-1,y}$ and~$b_{F,r-1,y}$: If~$j\in U$ is not contained in~$F$ then we do not want that a recolouring of~$v_j$ in the independent set allows a recolouring in the tree for~$F$. Thus, we use adjacency to the~$k$-clique to forbid the adjacent~$a$- or~$b$-vertex from taking colour~$4$ (which is exactly the one recolouring option that would have been possible by recolouring~$v_j$). Formally, if~$j$ is odd then we make~$a_{F,r-1,(j+1)/2}$ adjacent to vertex~$u_4$ of the~$k$-clique (forbidding colour~$4$ for~$a_{F,r-1,(j+1)/2}$), and if~$j$ is even then we make~$b_{F,r-1,j/2}$ adjacent to~$u_4$.

This completes the construction of the graph~$G$. We have already specified colours under~$\alpha$ and~$\beta$ for all vertices, and we recall that only~$s$ and~$t$ have different colours with respect to $\alpha$ and $\beta$. The necessary recolouring of~$t$ to~$4$, however, will cause a substantial number of recolourings and, along the way, capture the selection of a hitting set for~$\F$. We define the maximum number~$\ell$ of recolouring steps as
\[
\ell=3+2p+2m\cdot 3\log n.
\]
Intuitively, this value is intended as follows: (1) $p$ recolourings for the vertices in the independent set that correspond to a~$p$-hitting set, (2) $m\cdot 3\log n$ recolourings to propagate the selected hitting set up to all vertices~$d_{F,0,1}$ (giving them colour~$1$), (3) three recolourings for~$s$ and~$t$, namely~$t\rightarrow 4$, $s\rightarrow 3$, and~$t\rightarrow 2$ (4) $m\cdot 3\log n$ recolourings to undo the hitting set propagation, and (5) $p$ recolourings to undo the selection of the hitting set. We return~$(G,\alpha,\beta,\ell)$ as the output of our transformation. Since~$p<m$ and~$\log n=O(m)$, we have that~$\ell$ is polynomially bounded in~$m$, in fact~$\ell=O(m^2)$, as claimed. Clearly, our transformation can be performed in polynomial time. It remains to prove correctness.

\medskip
\noindent
\emph{Completeness.} Assume that the initial instance~$(U,\F,m,p)$ of \textsc{Hitting Set($m$)} is yes and let~$S$ be a hitting set of size at most~$p$ for~$\F$. We outline a recolouring procedure following exactly the stated intuition for the budget~$\ell$. (All steps are strictly serial but we do not insist on an ordering if it is immaterial.):
\begin{enumerate}
 \item Recolour all~$v_j$ from~$4$ to~$1$ for all~$j\in S$. The only neighbours are~$a$- or~$b$-vertices that have colours~$2$ or~$3$. This uses at most~$p$ steps.
 \item For each~$F\in\F$, recolour bottom-up the vertices in the tree-like claw structure, beginning with some~$a$- or~$b$-vertex whose adjacent independent set vertex has been recoloured from~$4$ to~$1$. Since~$S$ is a hitting set for~$\F$, such a vertex can always be found.
 \begin{enumerate}
  \item Recolour the~$a$ (or~$b$) vertex from~$2$ (or~$3$) to~$4$; there is no conflict with the~$c$-vertex since that has colour~$1$, same as the adjacent independent set vertex.
  \item Recolour the~$c$-vertex from~$1$ to either~$2$ or~$3$; only one choice is possible depending on whether we previously recoloured the~$a$ or the~$b$-vertex.
  \item Recolour the~$d$-vertex from~$4$ to~$1$; there is no conflict with the~$c$-vertex of (now) colour~$2$ or~$3$.
 \end{enumerate}
 At this point the argument can be repeated since the recolouring of, say,~$d_{F,x,y}$ with~$x\geq 1$, to~$1$ permits a recolouring of~$a_{F,x-1,(y+1)/2}$ or~$b_{F,x-1,y/2}$ to~$4$ depending on the parity of~$y$. Ultimately, we end up with~$d_{F,0,1}$ getting colour~$1$ (which does not conflict with~$t$ being colour~$3$).
 
 Over all sets~$F$ this uses~$m\cdot 3\log n$ steps since the tree arrangement has height~$r=\log n$ and we recolour three vertices in each claw.
 \item We then recolour~$t$ to~$4$,~$s$ to~$3$, and~$t$ to~$2$. This costs three steps and fulfills the requirement of~$\beta$ for both vertices. Clearly there are no conflicts.
 \item We then trace back the recolourings in each tree structure, using~$m\cdot 3\log n$ steps, followed by undoing the recolourings on vertices~$v_j$ corresponding to the hitting set~$S$, using at most~$p$ steps. This meets the requirement~$\beta$ for all vertices other than~$s$ and~$t$. (Note that~$d_{F,0,1}$ changing back from~$1$ to~$4$ makes no conflicts with~$t$ which now has colour~$2$.)
\end{enumerate}
Overall we obtain a recolouring sequence of length at most~$\ell$, as claimed. Thus $(G,\alpha,\beta,\ell)$ is indeed yes for $k$-\textsc{Colouring Reconfiguration}.

\medskip
\noindent
\emph{Soundness.} Let us assume that the obtained instance~$(G,\alpha,\beta,\ell)$ is yes for~$k$-\textsc{Colouring Reconfiguration} and let~$\alpha=\gamma_0,\ldots,\gamma_\ell=\beta$ be a sequence of proper recolourings. (We can repeat the last colouring in case that less than~$\ell$ colouring steps are needed.) We begin with some basic arguments about the behaviour of~$\gamma_0,\ldots,\gamma_\ell$.

If~$t$ would never receive colour~$4$ in any~$\gamma_i$ then it can be easily seen that the sequence must be infeasible: Indeed, this would restrict~$s$ to colours~$2$ and~$3$, and~$t$ to colours~$2$ and~$3$. This makes it impossible to, effectively, swap the colours of~$s$ and~$t$ since they are adjacent. Thus, let~$z\in\{1,\ldots,\ell\}$ be the smallest integer such that~$\gamma_z(t)=4$. Clearly, since~$\beta(t)=2$ and~$\alpha(s)=2\neq 3=\beta(s)$ at least~$3$ recolouring steps address vertices~$s$ and~$t$, leaving at most~$2p+2m\cdot 3\log n$ for the remaining vertices.

Since~$\alpha(d_{F,0,1})=4$ for all~$F\in\F$ it follows that each~$d_{F,0,1}$ must be recoloured to a colour other than~$4$ before step~$z$ (in which~$t$ for the first time is coloured~$4$), and recall that the adjacency to the~$k$-clique forbids all colours other than~$1$ and~$4$. As discussed earlier, this propagates the need for earlier recolourings through each tree-like arrangement of claws. Ultimately, for each~$F\in\F$ at least one~$a$- or one~$b$-vertex with index~$(F,r-1,\cdot)$ must be recoloured before step~$z$, along with a total~$3\log n$ vertices in that tree. Since all these vertices need to be reverted to their original colours later, this leaves only a budget of at most~$2p$ for the independent set vertices.

In the independent set, there are forced recolourings from~$4$ to~$1$ for all~$v_j$ that have 
adjacent recoloured~$a$- or~$b$-vertices with index~$(\cdot,r-1,\cdot)$. 
 Since all these must return to colour~$4$ in~$\gamma_\ell=\beta$ latest, at most~$p$ of these vertices can ever be recoloured. Let~$S\subseteq U$ denote the elements of~$U$ that correspond to these vertices. We will prove that~$S$ is a hitting set for~$\F$; clearly~$|S|\leq p$.

Fix any set~$F\in \F$. We already argued that a recolouring of~$d_{F,0,1}$ from~$4$ to~$1$ ultimately requires a prior recolouring of some~$a$- or~$b$-vertex with index~$(F,r-1,\cdot)$. Recall, however, that we disallowed such recolourings whenever the corresponding element~$j\in U$ is not contained in~$F$. (Here, corresponding refers to our construction where~$a_{F,r-1,y}$ is adjacent to~$v_{2y-1}$ and~$b_{F,r-1,y}$ is adjacent to~$v_{2y}$.) Thus, we must have recoloured an~$a$- or~$b$-vertex with index~$(F,r-1,y)$ such that the corresponding element~$j\in\{2y-1,2y\}$ is contained in~$F$. This in turn, as discussed earlier, requires a prior recolouring of~$v_j$ which implies that~$j\in S$. Thus,~$S$ indeed has a nonempty intersection with~$F$, implying that~$S$ is a hitting set for~$\F$, as claimed.\qed

\section{Conclusions} \label{sec:concl}

We showed that 
{\sc $k$-Colouring Reconfiguration} is 
fixed-parameter tractable for any fixed $k\geq 1$, when  parameterized by the number of recolourings $\ell$. 
It is a natural question to ask whether a single-exponential \FPT\ algorithm
can be achieved for this problem. 
We also proved that the $k$-\textsc{Colouring Reconfiguration} problem is polynomial-time solvable for $k=3$, which solves the open problem of Cereceda et al.~\cite{CHJ06b}, and
that it has no polynomial kernel for all $k\geq 4$, when parameterized by~$\ell$ (up to 
the standard complexity assumption that
$\mathsf{NP\nsubseteq coNP/poly}$).

\medskip
\noindent
{\it Acknowledgements.} 
We are grateful to several reviewers for insightful comments that greatly improved our presentation.


\begin{thebibliography}{10}
\setcounter{bibcounter}{\value{enumiv}}

\bibitem{BJK14}
H.~L. Bodlaender, B.M.P. Jansen and S. Kratsch, Kernelization lower bounds by cross-composition, SIAM Journal on Discrete Mathematics 28 (2014) 277--305.

\bibitem{BB13}
M. Bonamy and N. Bousquet, Recoloring bounded treewidth graphs,  Electronic Notes in Discrete Mathematics 44 (2013) 257--262.

\bibitem{BJLPP14}
M. Bonamy, M. Johnson, I.M. Lignos, V. Patel and D. Paulusma, Reconfiguration graphs for vertex colourings of chordal and chordal bipartite graphs, Journal of Combinatorial Optimization 27 (2014) 132--143.

\bibitem{Bo10}
P. Bonsma, The complexity of rerouting shortest paths, In B. Rovan, V. Sassone, P. Widmayer (eds.) Mathematical Foundations of Computer Science (MFCS 2012). Lecture Notes in Computer Science, vol. 7464, pp. 222--233. Springer Berlin Heidelberg (2012).

\bibitem{Bo14}
P. Bonsma, Independent set reconfiguration in cographs, WG 2014, Lecture Notes in Computer Science, to appear.

\bibitem{Bo13}
P. Bonsma, Rerouting shortest paths in planar graphs,
In D. D'Souza, T. Kavitha, J. Radhakrishnan (eds.) IARCS Annual Conference on Foundations of Software Technology and Theoretical Computer Science (FSTTCS 2012). LIPIcs, vol. 18, pp. 337--349. Schloss Dagstuhl--Leibniz-Zentrum f\"ur Informatik (2012).
  
\bibitem{BC09}
P. Bonsma and L. Cereceda, Finding paths between graph colourings: PSPACE-completeness and superpolynomial distances, 
Theoretical Computer Science
410 (2009) 5215--5226. 

\bibitem{BM14}
P. Bonsma, A. E. Mouawad, The complexity of bounded length graph recolouring, 
Manuscript (2014)
arXiv:1404.0337.

\bibitem{BKW14}
P. Bonsma, M. Kami\'nski, M. Wrochna, Reconfiguring independent sets in claw-free graphs, 
In 14th Scandinavian Symposium and Workshops on Algorithm Theory (SWAT 2014), 
Lecture Notes in Computer Science, vol. 8503, pp. 86--97. Springer Berlin Heidelberg (2014).

\bibitem{CHJ06}
L. Cereceda, J.~van den Heuvel and M.~Johnson,
Connectedness of the graph of vertex-colourings, Discrete
Mathematics 308 (2008) 913--919.

\bibitem{CHJ06a}
L. Cereceda, J.~van den Heuvel and M.~Johnson, Mixing 3-colourings in bipartite graphs, European Journal of Combinatorics 30 (2009) 1593--1606. 

\bibitem{CHJ06b}
L. Cereceda, J.~van den Heuvel and  M.~Johnson,
Finding paths between 3-colourings, Journal of Graph Theory 67 (2010) 69--82.

\bibitem{DLS09}
Michael Dom, Daniel Lokshtanov, and Saket Saurabh, Incompressibility through colors and ids, Proc. ICALP 2009, 
Lecture Notes in Computer Science 5555 (2009) 378--389.

\bibitem{GKMP06} 
P.~Gopalan, P.~G.~Kolaitis, E.~N.~Maneva and C.~H.~Papadimitriou,  The connectivity of boolean satisfiability: computational
and structural dichotomies, SIAM Journal on Computing 38 (2009) 2330--2355.

\bibitem{He13}
J. van den Heuvel, The complexity of change, 
Surveys in Combinatorics 2013, 
London Mathematical Society Lecture Notes Series 409.

\bibitem{IDHPSUU10}
T.~Ito, E.~D.~Demaine, N.~J.~A.~Harvey, C.~H.~Papadimitriou, M.~Sideri, R.~Uehara and Y.~ Uno, On the complexity of reconfiguration problems, Theoretical Computer Science 412 (2010) 1054--1065. 

\bibitem{IKD09}
T.~Ito, M.~Kami\'nski and E.~D.~Demaine, Reconfiguration of list edge-colorings in a graph, Proc. WADS 2009,
Lecture Notes in Computer Science 
5664 (2009) 375--386.

\bibitem{IKOZ12}
T. Ito, K. Kawamura, H. Ono, X. Zhou, Reconfiguration of list $L(2, 1)$-labelings in a graph, In K-M. Chao, T-S. Hsu, D-T. Lee (eds.)  Algorithms and Computation (ISAAC 2012). Lecture Notes in Computer Science, vol. 7676, pp. 34--43. Springer Berlin Heidelberg (2012).

\bibitem{IKZ11}
T.~Ito, K~Kawamura, X.~Zhou,
An improved sufficient condition for reconfiguration of list edge-colorings in a tree, In M. Ogihara, J. Tarui (eds.)  Theory and Applications of Models of Computation (TAMC 2011). Lecture Notes in Computer Science, vol. 6648, pp. 94--105. Springer Berlin Heidelberg (2011).

\bibitem{ID11}
T.~Ito, E.~D.~Demaine,
Approximability of the subset sum reconfiguration problem, 
In M. Ogihara, J. Tarui (eds.)  Theory and Applications of Models of Computation (TAMC 2011). Lecture Notes in Computer Science, vol. 6648, pp. 58--69. Springer Berlin Heidelberg (2011).

\bibitem{KMM12}
M.~Kami\'nski, P.~Medvedev and M.~Milani\v{c},
Complexity of independent set reconfigurability problems,
Theoretical Computer Science 439 (2012) 9--15.

\bibitem{KMM11}
M.~Kami\'nski, P.~Medvedev and M.~Milani\v{c},
Shortest paths between shortest paths,  Theoretical Computer Science 412 (2011) 5205--5210. 

\bibitem{Lo73}
L. Lov\'asz, Coverings and coloring of hypergraphs, Proc. 4th  Southeastern  Conference on Combinatorics, Graph Theory, and Computing, Utilitas Math. (1973) 3--12.

\bibitem{MNR14}
A. E. Mouawad, N. Nishimura and V. Raman, Vertex cover reconfiguration and beyond, Manuscript (2014) arXiv:1402.4926.

\bibitem{MNRSS13}
A. E. Mouawad, N. Nishimura, V. Raman, N. Simjour,  A. Suzuki, On the parameterized complexity of reconfiguration problems, 
In G. Gutin,  S. Szeider (eds.)  Parameterized and Exact Computation (IPEC 2013). Lecture Notes in Computer Science, vol. 8246, pp. 281--294. Springer Berlin Heidelberg (2013).

\bibitem{Ni06}
R.~Niedermeier, Invitation to fixed-parameter algorithms, 
vol.~31 of Oxford Lecture Series in Mathematics and its Applications, Oxford University Press, 
Oxford, 2006.

\setcounter{bibcounter}{\value{enumiv}}
\end{thebibliography}
\end{document}